\documentclass[journal]{IEEEtran}
\usepackage[cmex10]{amsmath}
\usepackage{graphicx,amssymb,amsthm,comment,bm,cite,url}
\usepackage{color}
\usepackage{boldline}
\usepackage{algorithm,algorithmic}
\usepackage{multirow}
\allowdisplaybreaks

\newcommand{\refeq}[1]{(\ref{eq:#1})}
\newcommand{\refeqs}[2]{(\ref{eq:#1}) and (\ref{eq:#2})}

\newcommand{\refsec}[1]{Section \ref{sec:#1}}
\newcommand{\refsubsec}[1]{\ref{subsec:#1}}

\newcommand{\reffig}[1]{Fig. \ref{fig:#1}}
\newcommand{\reffigs}[2]{Figs. \ref{fig:#1} and \ref{fig:#2}}
\newcommand{\reffigss}[2]{Figs. \ref{fig:#1}--\ref{fig:#2}}
\newcommand{\reftab}[1]{Table \ref{tab:#1}}
\newcommand{\reftabs}[2]{Tables \ref{tab:#1} and \ref{tab:#2}}

\newtheorem{theorem}{Theorem}
\newtheorem{proposition}{Proposition}

\def\Vec#1{\boldsymbol{\mathbf{#1}}}

\def\thline{\noalign{\hrule height 1.2pt}}

\ifCLASSINFOpdf
\else
\fi

\begin{document}
\title{
Nonparallel Voice Conversion with 
Augmented Classifier 
Star Generative Adversarial Networks
}

\author{Hirokazu~Kameoka,~
Takuhiro Kaneko,
Kou Tanaka,
and Nobukatsu Hojo%
\thanks{H. Kameoka, T. Kaneko, K. Tanaka, 
and N. Hojo are with NTT Communication Science Laboratories, Nippon Telegraph and Telephone Corporation, Atsugi, Kanagawa, 243-0198 Japan (e-mail: hirokazu.kameoka.uh@hco.ntt.co.jp).}
\thanks{
This work was supported by 
JSPS KAKENHI 17H01763 and JST CREST Grant Number JPMJCR19A3, Japan. 
}}

\markboth{}%
{}

\maketitle

\begin{abstract}
We previously proposed a method that allows for nonparallel voice conversion (VC) by using a variant of generative adversarial networks (GANs) called StarGAN. The main features of our method, called StarGAN-VC, are as follows: First, it requires no parallel utterances, transcriptions, or time alignment procedures for speech generator training. Second, it can simultaneously learn mappings across multiple domains using a single generator network and thus fully exploit available training data collected from multiple domains to capture latent features that are common to all the domains. Third, it can generate converted speech signals quickly enough to allow real-time implementations and requires only several minutes of training examples to generate reasonably realistic-sounding speech. In this paper, we describe three formulations of StarGAN, including a newly introduced novel StarGAN variant called ``Augmented classifier StarGAN (A-StarGAN)'', and compare them in a nonparallel VC task. We also compare them with several baseline methods. 
\end{abstract}


\IEEEpeerreviewmaketitle

\section{Introduction}
\label{sec:intro}

Voice conversion (VC) is a task of converting the voice of
a source speaker without changing the uttered sentence. 
Examples of the applications of VC techniques include 
speaker-identity modification \cite{Kain1998short}, speaking assistance \cite{Kain2007short,Nakamura2012short}, 
speech enhancement \cite{Inanoglu2009short,Turk2010short,Toda2012short}, 
bandwidth extension \cite{Jax2003short},
and accent conversion \cite{Felps2009short}.

One successful VC framework involves 
approaches that utilize acoustic models represented by Gaussian mixture models (GMMs) for feature mapping \cite{Stylianou1998short,Toda2007short,Helander2010short}.
Recently, frameworks based on neural networks (NNs)  
\cite{Desai2010short,Mohammadi2014short,Sun2015short,YSaito2017bshort,Kaneko2017cshort,Chen2014short,Nakashika2014ashort,Nakashika2014bshort,Nakashika2015short,Blaauw2016short,Hsu2016short,Hsu2017short,Xie2016short,Kinnunen2017short,Kaneko2017dshort,vandenOord2017bshort,Hashimoto2017short,YSaito2018bshort,Kameoka2019IEEETransshort_ACVAE-VC} and an exemplar-based framework based on nonnegative matrix factorization (NMF) \cite{Takashima2013short,Wu2014short,Sisman2019short_nmf} 
have also proved successful. 
Many conventional VC methods, including those mentioned above, require accurately aligned parallel 
source and target speech data. However, in many scenarios, 
it is not always possible to collect parallel utterances. 
Even if we could collect such data, we typically need to perform 
time alignment procedures, which becomes relatively
difficult when there is a large acoustic gap between
the source and target speech. Since many frameworks are weak as regards the misalignment
found with parallel data, careful pre-screening and
manual correction may be required to make these frameworks work reliably.
To bypass these restrictions, this paper is concerned with developing a nonparallel  
VC method, which requires no parallel utterances, transcriptions, or time alignment procedures.

In general, the quality and conversion effect obtained with nonparallel methods
are usually limited compared with methods using parallel data
due to the disadvantage related to the training condition. 
Thus, developing nonparallel methods whose speech quality and a conversion effect are as high as those of parallel methods can be very challenging. 
Recently, some attempts have been made to develop nonparallel methods \cite{Chen2014short,Nakashika2014ashort,Nakashika2014bshort,Nakashika2015short,Blaauw2016short,Hsu2016short,Hsu2017short,Xie2016short,Kinnunen2017short,Kaneko2017dshort,vandenOord2017bshort,Hashimoto2017short,YSaito2018bshort,Kameoka2019IEEETransshort_ACVAE-VC}.
One example is a method using automatic speech recognition (ASR) \cite{Xie2016short}.
The idea is to convert input speech under the restriction that the posterior state probability of the acoustic model of an ASR system is preserved so that the transcription of the converted speech becomes consistent with that of the input speech. 
Since the performance of this method depends heavily on the quality of the acoustic model of ASR,
it can fail to work if ASR does not function reliably. 
A method using i-vectors \cite{Dehak2011short}, known as a feature for speaker verification, was proposed in \cite{Kinnunen2017short}.
Conceptually, the idea is to
shift the acoustic features of input speech towards target speech in the i-vector space so that
the converted speech is likely to be recognized as the target speaker by a speaker recognizer.
While this method is also free from parallel data, one limitation is that 
it is applicable only to speaker identity conversion tasks.

Recently, a framework based on conditional variational autoencoders (CVAEs) \cite{Kingma2014ashort,Kingma2014bshort} was proposed in \cite{Hsu2016short,YSaito2018bshort,Kameoka2019IEEETransshort_ACVAE-VC}.
As the name implies, variational autoencoders (VAEs), 
consisting of encoder and decoder networks,
are probabilistic counterparts of autoencoders (AEs).
CVAEs \cite{Kingma2014bshort} are an extended version of VAEs where the encoder and decoder networks can take a class indicator variable 
as an additional input. 
By using acoustic features as the training examples and 
the associated domain class labels, 
the networks learn how to convert source speech to a target domain 
according to the domain class label fed into the decoder. 
This CVAE-based VC approach is notable in that it is completely free from parallel data and works even with unaligned corpora. 
However, one well-known problem as regards VAEs is that outputs from the decoder tend to be oversmoothed. 
For VC applications, this can be problematic since it usually results in poor quality buzzy-sounding speech. 

One powerful framework that can potentially overcome the weakness of VAEs involves generative adversarial networks (GANs) \cite{Goodfellow2014short}.
GANs offer a general framework 
for training a generator network 
so that it can generate fake data samples 
that can deceive a real/fake discriminator network 
in the form of a minimax game. 
While they have been found to be effective for use with image generation, 
in recent years they 
have also been employed with notable success for various speech processing tasks \cite{Kaneko2017ashort,YSaito2018ashort,Pascual2017short,Kaneko2017bshort,Kaneko2017cshort,Oyamada2018short}.
We previously reported a nonparallel VC method using a GAN variant called cycle-consistent GAN (CycleGAN) \cite{Kaneko2017dshort}, which
was originally proposed as a method for translating images using unpaired training examples \cite{Zhu2017short,Kim2017short,Yi2017short}. 
Although this method, which we call CycleGAN-VC, was shown to work reasonably well, 
one major limitation is that it only learns mappings between a single pair of domains. 
In many VC application scenarios, it is desirable to 
be able to convert speech into multiple domains, not just one. 
One naive way of applying CycleGAN to multi-domain VC tasks would be to prepare and train a 
different mapping 
pair for each domain pair. 
However, this can be ineffective 
since each mapping pair
fails to use the training data of the other domains for learning,
even though there must be a common set of latent features that can be shared 
across different domains.

To overcome the shortcomings and limitations of CVAE-VC \cite{Hsu2016short} and CycleGAN-VC \cite{Kaneko2017dshort}, we previously proposed a nonparallel VC method 
\cite{Kameoka2018SLTshort_StarGAN-VC} using
another GAN variant called StarGAN \cite{Choi2017short}, which 
offers the advantages of CVAE-VC and CycleGAN-VC concurrently. 
Unlike CycleGAN-VC and as with CVAE-VC, our method, which we call StarGAN-VC, is 
capable of simultaneously learning multiple mappings using a single generator network 
and can thus fully
use available training data collected from multiple domains.
Unlike CVAE-VC and as with CycleGAN-VC, StarGAN-VC uses an adversarial loss 
for generator training 
to encourage 
the generator outputs to become indistinguishable from real speech.
It is also noteworthy that unlike CVAE-VC and CycleGAN-VC, 
StarGAN-VC does not require any information about the domain of the input speech at test time.

In this paper, we describe three formulations of StarGAN, including a newly introduced novel StarGAN variant called ``Augmented classifier StarGAN'', and compare them in a nonparallel VC task.
The remainder of this paper is organized as follows.
After reviewing other related work in \refsec{relatedwork},
we briefly describe the formulation of CycleGAN-VC in \refsec{cyclegan-vc},
present the three formulations of StarGAN-VC in \refsec{stargan-vc},
and show experimental results in \refsec{experiments}.

\section{Related Work}
\label{sec:relatedwork}

Other natural ways of overcoming the weakness of VAEs includes
the VAE-GAN framework \cite{Larsen2015short}.
A nonparallel VC method based on this framework has already been proposed in \cite{Hsu2017short}.
With this approach, an adversarial loss derived using a GAN discriminator is incorporated into the training loss to encourage the decoder outputs of a CVAE to be indistinguishable from real speech features. 
Although the concept is similar to our StarGAN-VC approach, 
we will show in \refsec{experiments} that our approach outperforms this method 
in terms of both speech quality and the conversion effect.

Another related technique worth noting is 
the vector quantized VAE (VQ-VAE) approach \cite{vandenOord2017bshort}, 
which has performed impressively in nonparallel VC tasks. 
This approach 
is particularly notable in that it offers a novel way of overcoming the weakness of VAEs by using
the WaveNet model \cite{vandenOord2016short}, a sample-by-sample neural signal generator, to devise both the encoder and decoder of a discrete counterpart of CVAEs.
The original WaveNet model is a recursive model that makes it possible to predict the distribution of a sample conditioned on the samples the generator has produced. 
While a faster version \cite{vandenOord2017ashort} has recently been proposed,
it typically requires huge computational cost to generate a stream of samples, which can cause difficulties when implementing real-time systems.
The model is also known to require a huge number of training examples to generate natural-sounding speech.
By contrast, our method is noteworthy in that it is able to generate signals quickly enough to allow real-time implementation and requires only several minutes of training examples to generate reasonably realistic-sounding speech.

Meanwhile, 
given the recent success of the sequence-to-sequence (S2S) learning framework
in various tasks, 
several VC methods based on S2S models have been
proposed, including the ones we proposed previously \cite{Tanaka2019short,Kameoka2020_ConvS2S-VC,Huang2019arXiv_VTN,Kameoka2020_VTN}. 
While S2S models usually require parallel corpora for training, 
an attempt has also been made to train an S2S model using nonparallel utterances 
\cite{Zhang2019}.
However, 
it requires phoneme transcriptions as auxiliary information
for model training.

\section{CycleGAN Voice Conversion}
\label{sec:cyclegan-vc}

Since StarGAN-VC is an extension of CycleGAN-VC, which we proposed previously \cite{Kaneko2017dshort},  
we start by briefly reviewing its formulation (\reffig{cyclegan}).

Let $\Vec{x}\in \mathbb{R}^{Q\times N}$ and $\Vec{y}\in\mathbb{R}^{Q\times M}$
be acoustic feature sequences of speech belonging to domains $X$ and $Y$, respectively,
where $Q$ is the feature dimension and $N$ and $M$ are the lengths of the sequences.
In the following, we will restrict our attention to speaker identity conversion tasks,
so when we use the term domain, we will mean speaker.
The aim of CycleGAN-VC is to learn 
a mapping $G$ that converts the domain of $\Vec{x}$ into $Y$
and a mapping $F$ that does the opposite. 
Now, we introduce discriminators $D_X$ and $D_Y$,
whose roles are to predict whether or not their inputs are the acoustic features of real speech belonging to $X$ and $Y$,
and define 
\begin{align}
\mathcal{L}_{\rm adv}^{D_Y}(D_Y) 
=&- 
\mathbb{E}_{\Vec{y}\sim p_Y(\Vec{y})}[\log D_Y(\Vec{y})]\nonumber\\
&-\mathbb{E}_{\Vec{x}\sim p_X(\Vec{x})}[\log (1- D_Y(G(\Vec{x})))],
\label{eq:cyclegan-advloss_dy}
\\
\mathcal{L}_{\rm adv}^{G}(G) 
=&\mathbb{E}_{\Vec{x}\sim p_X(\Vec{x})}[\log (1- D_Y(G(\Vec{x})))],
\label{eq:cyclegan-advloss_g}
\\
\mathcal{L}_{\rm adv}^{D_X}(D_X) 
=&-
\mathbb{E}_{\Vec{x}\sim p_X(\Vec{x})}[\log D_X(\Vec{x})]\nonumber\\
&-
\mathbb{E}_{\Vec{y}\sim p_Y(\Vec{y})}[\log (1- D_X(F(\Vec{y})))],
\label{eq:cyclegan-advloss_dx}
\\
\mathcal{L}_{\rm adv}^{F}(F) 
=&
\mathbb{E}_{\Vec{y}\sim p_Y(\Vec{y})}[\log (1- D_X(F(\Vec{y})))],
\label{eq:cyclegan-advloss_f}
\end{align}
as the adversarial losses for $D_Y$, $G$, $D_X$ and $F$, 
respectively.
$\mathcal{L}_{\rm adv}^{D_Y}(D_Y)$ 
and 
$\mathcal{L}_{\rm adv}^{D_X}(D_X) $
measure how indistinguishable 
$G(\Vec{x})$ and $F(\Vec{y})$ are 
from acoustic features of real speech belonging to $Y$ and $X$. 
Since the goal of $D_X$ and $D_Y$ is to correctly distinguish 
the converted feature sequences obtained via $G$ and $F$
from real speech feature sequences, 
$D_X$ and $D_Y$ attempt to minimize these losses to avoid being fooled by $G$ and $F$.
Conversely, since one of the goals of $G$ and $F$ is to generate natural-sounding speech that is indistinguishable from real speech, $G$ and $F$ attempt to maximize these losses or minimize 
$\mathcal{L}_{\rm adv}^{G}(G)$ and $\mathcal{L}_{\rm adv}^{F}(F)$ to fool $D_Y$ and $D_X$.
It can be shown that the output distributions of $G$ and $F$ trained in this way will match the empirical distributions 
$p_Y(\Vec{y})$ and $p_X(\Vec{x})$ 
if $G$, $F$, $D_X$, and $D_Y$ have enough capacity \cite{Zhu2017short,Goodfellow2014short}. 
Note that since 
$\mathcal{L}_{\rm adv}^{G}(G)$ and $\mathcal{L}_{\rm adv}^{F}(F)$ 
are minimized when $D_Y(G(\Vec{x}))= 1$ and $D_X(F(\Vec{y}))= 1$,
we can also use 
$- \mathbb{E}_{\Vec{x}\sim p_X(\Vec{x})}[\log D_Y(G(\Vec{x}))]$
and 
$- \mathbb{E}_{\Vec{x}\sim p_X(\Vec{x})}[\log D_Y(G(\Vec{x}))]$ 
as the adversarial losses for $G$ and $F$.

As mentioned above,
training $G$ and $F$ using the adversarial losses
enables mappings $G$ and $F$ to produce outputs identically distributed 
as target domains $Y$ and $X$, respectively.
However,  
using them alone does not guarantee that 
$G$ or $F$ will preserve the linguistic contents of 
input speech since there are infinitely many mappings 
that will induce the same output distributions.
One way to let $G$ and $F$ preserve the linguistic contents of input speech 
would be to encourage them to make only minimal changes from the inputs.
To incentivize this behaviour,
we introduce a cycle consistency loss \cite{Zhu2017short,Kim2017short,Yi2017short}
\begin{align}
\mathcal{L}_{\rm cyc}(G,F) 
&= \mathbb{E}_{\Vec{x}\sim p_X(\Vec{x})}[\| F(G(\Vec{x})) - \Vec{x} \|_\rho^\rho]
\nonumber\\
&+\mathbb{E}_{\Vec{y}\sim p_Y(\Vec{y})}[\| G(F(\Vec{y})) - \Vec{y} \|_\rho^\rho],
\end{align}
to enforce $F(G(\Vec{x}))\simeq \Vec{x}$ and $G(F(\Vec{y}))\simeq \Vec{y}$. 
In image-to-image translation tasks, 
this regularization loss contributes to enabling $G$ and $F$ to change only the textures and colors of input images while preserving the domain-independent contents.
However, the effect this loss would have on VC tasks was nontrivial.
Our previous work \cite{Kaneko2017dshort} was among the first  
to show that it enables $G$ and $F$ to change only the voice characteristics
of input speech while preserving the linguistic content.
This regularization technique has recently proved 
effective also in the VAE-based VC methods \cite{Tobing2019short}. 
With the same motivation, we also consider an identity mapping loss
\begin{align}
\mathcal{L}_{\rm id}(G,F) 
&= \mathbb{E}_{\Vec{x}\sim p_X(\Vec{x})}[\| F(\Vec{x}) - \Vec{x} \|_\rho^\rho]
\nonumber\\
&+\mathbb{E}_{\Vec{y}\sim p_Y(\Vec{y})}[\| G(\Vec{y}) - \Vec{y} \|_\rho^\rho],
\end{align}
to ensure that inputs to $G$ and $F$ 
are kept unchanged when the inputs already belong to $Y$ and $X$.
The full objectives of CycleGAN-VC to be minimized with respect to 
$G$, $F$, $D_X$, and $D_Y$ 
are thus given as 
\begin{align}
\mathcal{I}_{G,F}(G,F) =&
\lambda_{\rm adv}
\mathcal{L}_{\rm adv}^{G}(G) +
\lambda_{\rm adv}
\mathcal{L}_{\rm adv}^{F}(F) 
\nonumber\\
&+ 
\lambda_{\rm cyc}
\mathcal{L}_{\rm cyc}(G,F) + 
\lambda_{\rm id}
\mathcal{L}_{\rm id}(G,F),\\
\mathcal{I}_{D}(D_X,D_Y) =&
\mathcal{L}_{\rm adv}^{D_X}(D_X)
+
\mathcal{L}_{\rm adv}^{D_Y}(D_Y),
\end{align}
where 
$\lambda_{\rm adv}\ge 0$,
$\lambda_{\rm cyc}\ge 0$, and $\lambda_{\rm id}\ge 0$ are 
regularization parameters, which weigh the importance of 
the adversarial,
cycle consistency, and identity mapping losses.
In practice,
we alternately update  
$G$, $F$, $D_X$, and $D_Y$ one at a time
while keeping the others fixed.

\begin{figure*}[t!]
  \begin{minipage}[t]{.53\linewidth}
  \centering
  \centerline{\includegraphics[height=5.9cm]{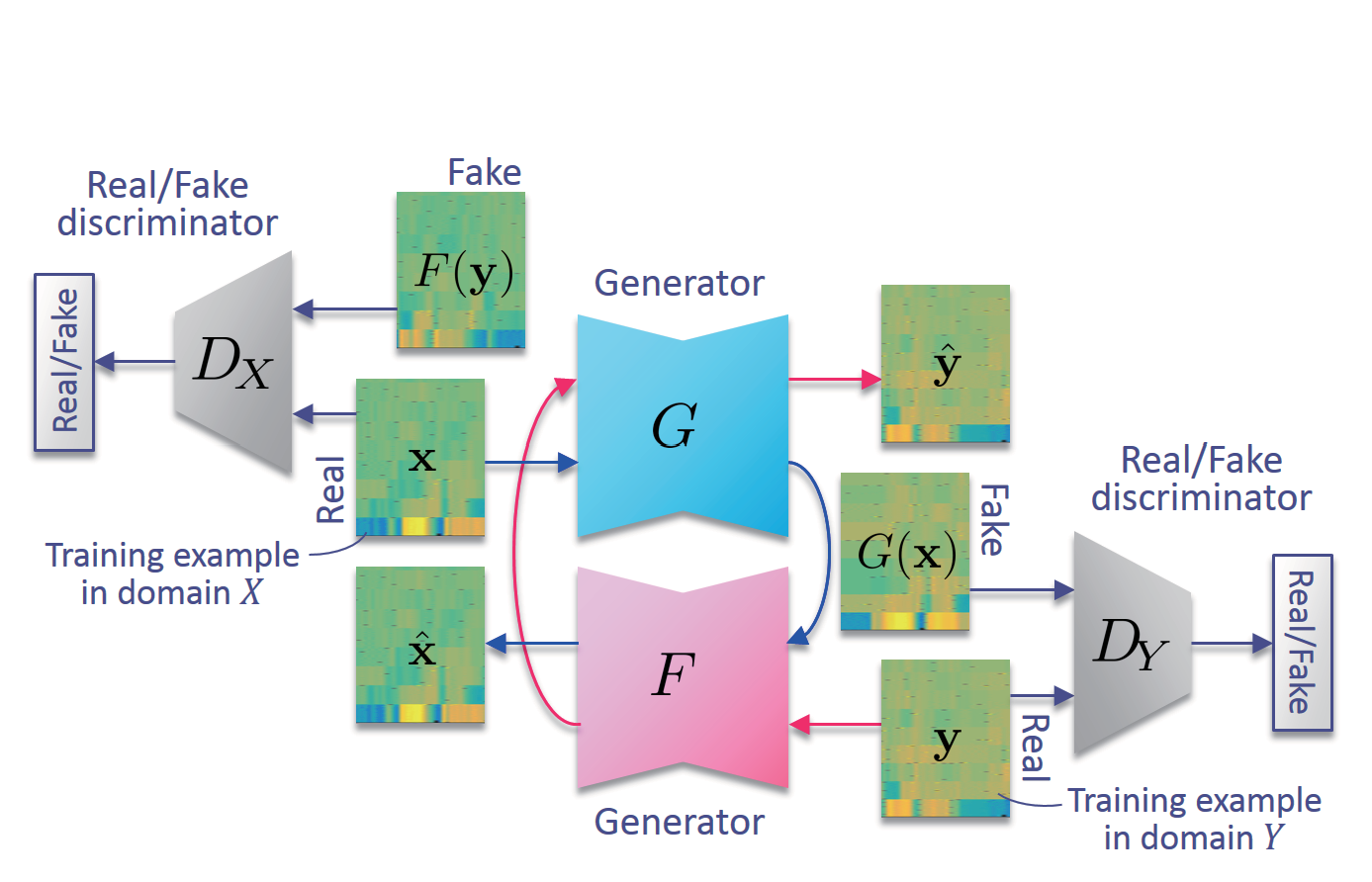}}
  \vspace{-1ex}
  \caption{Illustration of CycleGAN training.}
  \label{fig:cyclegan}
  \end{minipage}
  \hspace{2ex}
  \begin{minipage}[t]{.47\linewidth}
  \centering
  \centerline{\includegraphics[height=5.9cm]{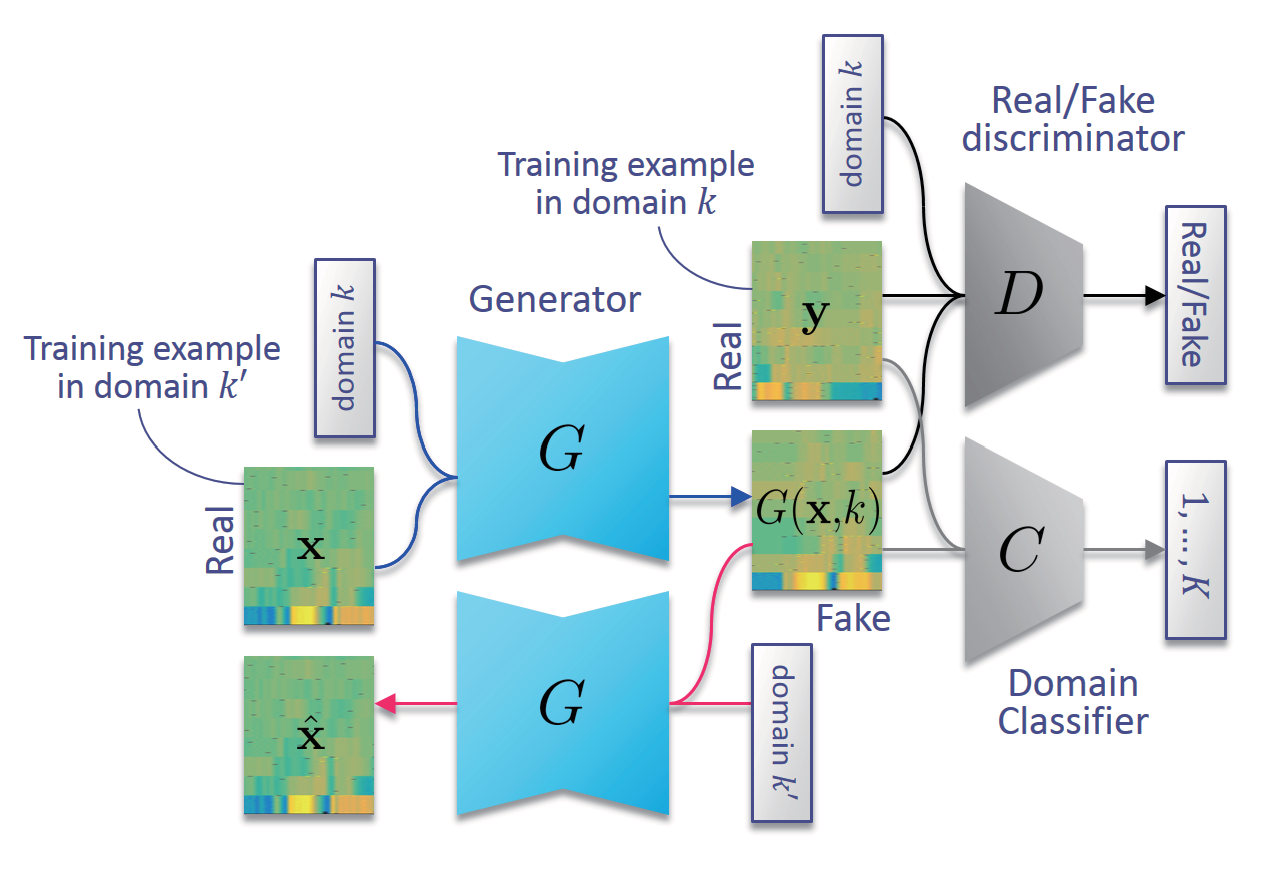}}
  \vspace{-1ex}
  \caption{Illustration of C-StarGAN training. 
  The $D$ network is designed to take the domain index $k$ as an additional input 
  and produce the probability of $\Vec{x}$ being a real data sample in domain $k$.}
  \label{fig:cstargan}
  \end{minipage}
\end{figure*}

\section{StarGAN Voice Conversion}
\label{sec:stargan-vc}

While CycleGAN-VC 
can only learn mappings between a single pair of domains, 
StarGAN-VC \cite{Kameoka2018SLTshort_StarGAN-VC} 
can learn mappings among multiple speech domains using a single generator network,
thus allowing us to fully utilize available training data collected from
multiple domains.
In this section, we describe three formulations of StarGAN. 
While the first and second formulations respectively correspond to the ones 
presented in \cite{Kameoka2018SLTshort_StarGAN-VC} and \cite{Choi2017short},
the third formulation is newly proposed in this paper 
with the aim of further improving the former two.

\subsection{Cross-Entropy StarGAN formulation}

First, we describe the formulation we introduced in \cite{Kameoka2018SLTshort_StarGAN-VC}.
Let $G$ be a generator that takes an acoustic feature sequence $\Vec{x}\in \mathbb{R}^{Q\times N}$ 
belonging to an arbitrary domain and a target domain class index $k\in\{1,\ldots,K\}$ as the inputs and generates an acoustic feature sequence $\hat{\Vec{y}} = G(\Vec{x},k)$. 
For example, if we consider speaker identities as the domain classes, 
each $k$ will be associated with a different speaker.
One of the goals of StarGAN-VC is to make 
$\hat{\Vec{y}} = G(\Vec{x},k)$ as realistic as real speech features
and belong to domain $k$. 
To achieve this, we introduce a real/fake discriminator $D$ as with CycleGAN 
and a domain classifier $C$, 
whose role is to predict to which classes an input belongs. 
$D$ is designed to produce a probability $D(\Vec{y},k)$ that an input $\Vec{y}$ is a real speech feature whereas 
$C$ is designed to produce class probabilities $p_C(k|\Vec{y})$ of $\Vec{y}$.

\noindent
{\bf Adversarial Loss:}
First, we define
\begin{align}
\mathcal{L}_{\rm adv}^D(D) =& 
- \mathbb{E}_{k \sim p(k), \Vec{y}\sim p_{\rm d}(\Vec{y}|k)}[\log D(\Vec{y},k)] 
\nonumber\\
&- \mathbb{E}_{k\sim p(k), \Vec{x}\sim p_{\rm d}(\Vec{x})}
[\log (1-D(G(\Vec{x},k),k))],
\label{eq:advloss_d}
\\
\mathcal{L}_{\rm adv}^G(G) =&
- \mathbb{E}_{k\sim p(k), \Vec{x}\sim p_{\rm d}(\Vec{x})}
[\log D(G(\Vec{x},k),k)],
\label{eq:advloss_g}
\end{align}
as adversarial losses for discriminator $D$ and generator $G$, respectively, 
where 
$p(k)$ is a uniform categorical distribution ($p(k)=\frac{1}{K}$),
$\Vec{y}\sim p_{\rm d}(\Vec{y}|k)$ denotes a training example of an acoustic feature sequence of real speech in domain $k$, and 
$\Vec{x}\sim p_{\rm d}(\Vec{x})$ denotes that in an arbitrary domain. 
$\mathcal{L}_{\rm adv}^D(D)$ takes a small value when $D$ correctly classifies 
$G(\Vec{x},k)$ and $\Vec{y}$ as fake and real speech features 
whereas
$\mathcal{L}_{\rm adv}^G(G)$ takes a small value when $G$ successfully deceives $D$ so that 
$G(\Vec{x},k)$ is misclassified as real speech features by $D$.
Thus, we would like to minimize $\mathcal{L}_{\rm adv}^D(D)$ with respect to $D$ and
minimize $\mathcal{L}_{\rm adv}^G(G)$ with respect to $G$.
Note that 
$\mathbb{E}_{k\sim p(k)}[\cdot]$ is a simplified notation for
$\frac{1}{K}\sum_{k=1}^{K}(\cdot)$, and
when $k$ denotes a speaker index,
$\mathbb{E}_{\Vec{y}\sim p_{\rm d}(\Vec{y}|k)}[\cdot]$
and
$\mathbb{E}_{\Vec{x}\sim p_{\rm d}(\Vec{x})}[\cdot]$
denote 
the sample means over the training examples of speaker $k$
and all speakers, respectively.
Note also that 
the order of the variables over which each expectation is taken
corresponds to the order of the for-loop in our implementation.

\noindent
{\bf Domain Classification Loss:} 
Next, we define
\begin{align}
\mathcal{L}_{\rm cls}^C(C) = &
- \mathbb{E}_{k\sim p(k), \Vec{y}\sim p_{\rm d}(\Vec{y}|k)}
[\log p_C(k|\Vec{y})],
\label{eq:clsloss_c}
\\
\mathcal{L}_{\rm cls}^G(G) = &
- \mathbb{E}_{k\sim p(k), \Vec{x}\sim p(\Vec{x})}
[\log p_C(k|G(\Vec{x},k))],
\label{eq:clsloss_g}
\end{align}
as domain classification losses for classifier $C$ and generator $G$.
$\mathcal{L}_{\rm cls}^C(C)$ and $\mathcal{L}_{\rm cls}^G(G)$ take small values
when $C$ correctly classifies $\Vec{y}\sim p_{\rm d}(\Vec{y}|k)$ and 
$G(\Vec{x},k)$ as belonging to domain $k$. 
Thus, we would like to minimize $\mathcal{L}_{\rm cls}^C(C)$ with respect to $C$ 
and $\mathcal{L}_{\rm cls}^G(G)$ with respect to $G$.

\noindent
{\bf Cycle Consistency Loss:} 
Training $G$, $D$, and $C$ using only the losses presented above does not guarantee that 
$G$ will preserve the linguistic content of input speech.  
As with CycleGAN-VC,
we introduce a cycle consistency loss to be minimized
\begin{multline}
\mathcal{L}_{\rm cyc}(G) \\
= 
\mathbb{E}_{k\sim p(k), k'\sim p(k), \Vec{x}\sim p_{\rm d}(\Vec{x}|k')}
[\| G(G(\Vec{x},k),k') - \Vec{x}\|_\rho^\rho],
\end{multline}
to encourage $G(\Vec{x},k)$ to
preserve the linguistic content of \Vec{x},
where 
$\Vec{x}\sim p_{\rm d}(\Vec{x}|k')$ denotes a training example of real speech feature sequences 
in domain $k'$, and $\rho$ is a positive constant.
We also consider an identity mapping loss
\begin{align}
\mathcal{L}_{\rm id}(G) = 
\mathbb{E}_{k'\sim p(k), \Vec{x}\sim p_{\rm d}(\Vec{x}|k')}
[\|G(\Vec{x},k') - \Vec{x} \|_\rho^\rho],
\end{align}
to ensure that an input into $G$ will remain unchanged when the input already belongs to domain $k'$.

To summarize, 
the full objectives  
to be minimized with respect to $G$, $D$, and $C$
are given as 
\begin{align}
\mathcal{I}_G(G) =
&\lambda_{\rm adv}
\mathcal{L}_{\rm adv}^{G}(G)
+
\lambda_{\rm cls}
\mathcal{L}_{\rm cls}^{G}(G)
\nonumber\\
&+
\lambda_{\rm cyc}
\mathcal{L}_{\rm cyc}(G)
+
\lambda_{\rm id}
\mathcal{L}_{\rm id}(G),
\\
\mathcal{I}_D(D) =
&\lambda_{\rm adv}
\mathcal{L}_{\rm adv}^{D}(D),
\\
\mathcal{I}_C(C) =
&\lambda_{\rm cls}
\mathcal{L}_{\rm cls}^{C}(C),
\end{align}
respectively, 
where $\lambda_{\rm adv}\ge 0$, 
$\lambda_{\rm cls}\ge 0$, $\lambda_{\rm cyc}\ge 0$ and $\lambda_{\rm id}\ge 0$ are 
regularization parameters, which weigh the importance of 
the adversarial,
domain classification,
cycle consistency, and 
identity mapping losses.
Since the adversarial and domain classification losses in 
(\ref{eq:advloss_d}), (\ref{eq:advloss_g}), (\ref{eq:clsloss_c}) and (\ref{eq:clsloss_g})
are defined using cross-entropy measures, 
we refer to this version of StarGAN as ``C-StarGAN'' (\reffig{cstargan}).

\begin{figure*}[t!]
  \begin{minipage}[t]{.47\linewidth}
  \centering
  \centerline{\includegraphics[height=5.9cm]{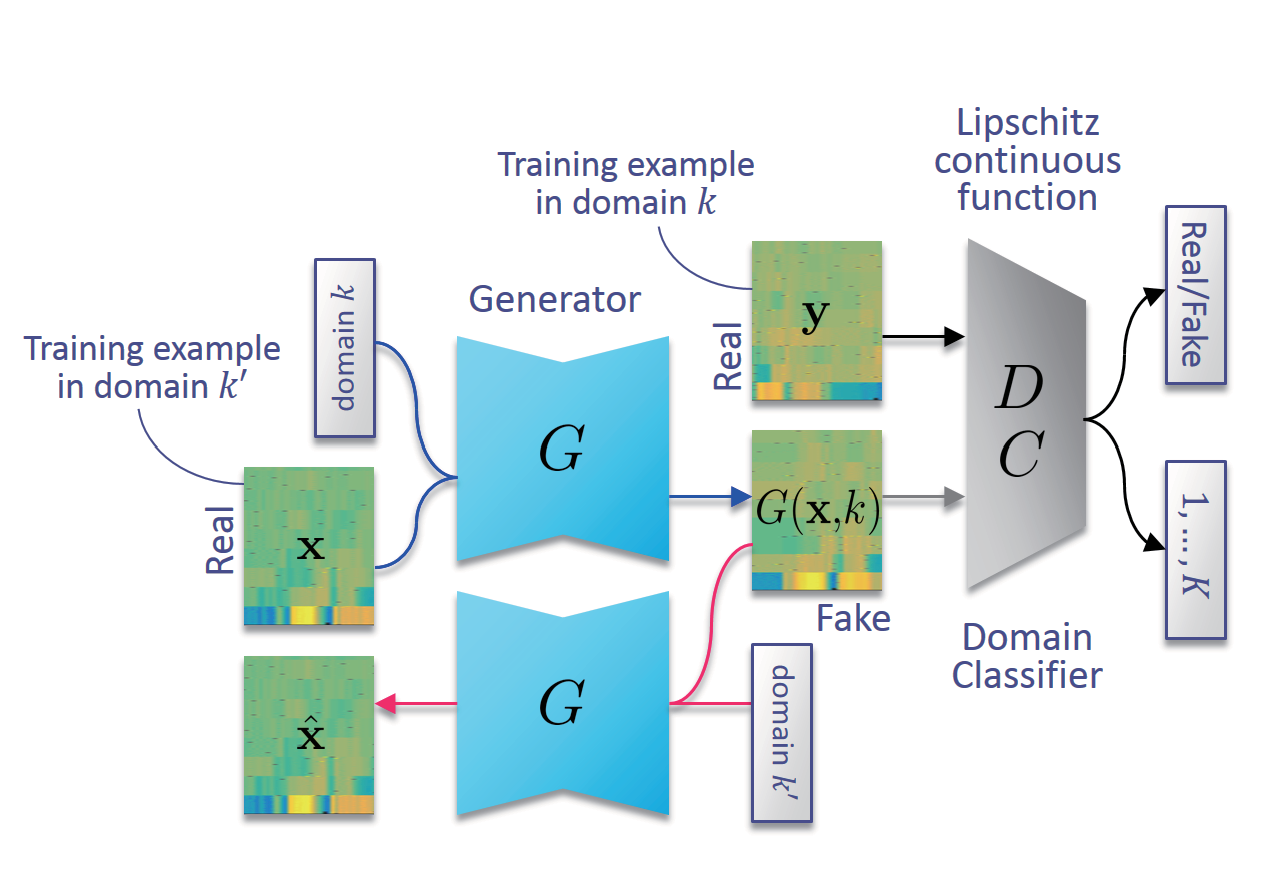}}
  \vspace{-1ex}
  \caption{Illustration of W-StarGAN training. 
  The $D$ and $C$ networks are designed to share lower layers
  and produce the score that measures how likely $\Vec{x}$ is to be a real data sample 
  and the probability of $\Vec{x}$ belonging to each domain.}
  \label{fig:wstargan}
  \end{minipage}
  \hspace{2ex}
  \begin{minipage}[t]{.47\linewidth}
  \centering
  \centerline{\includegraphics[height=5.9cm]{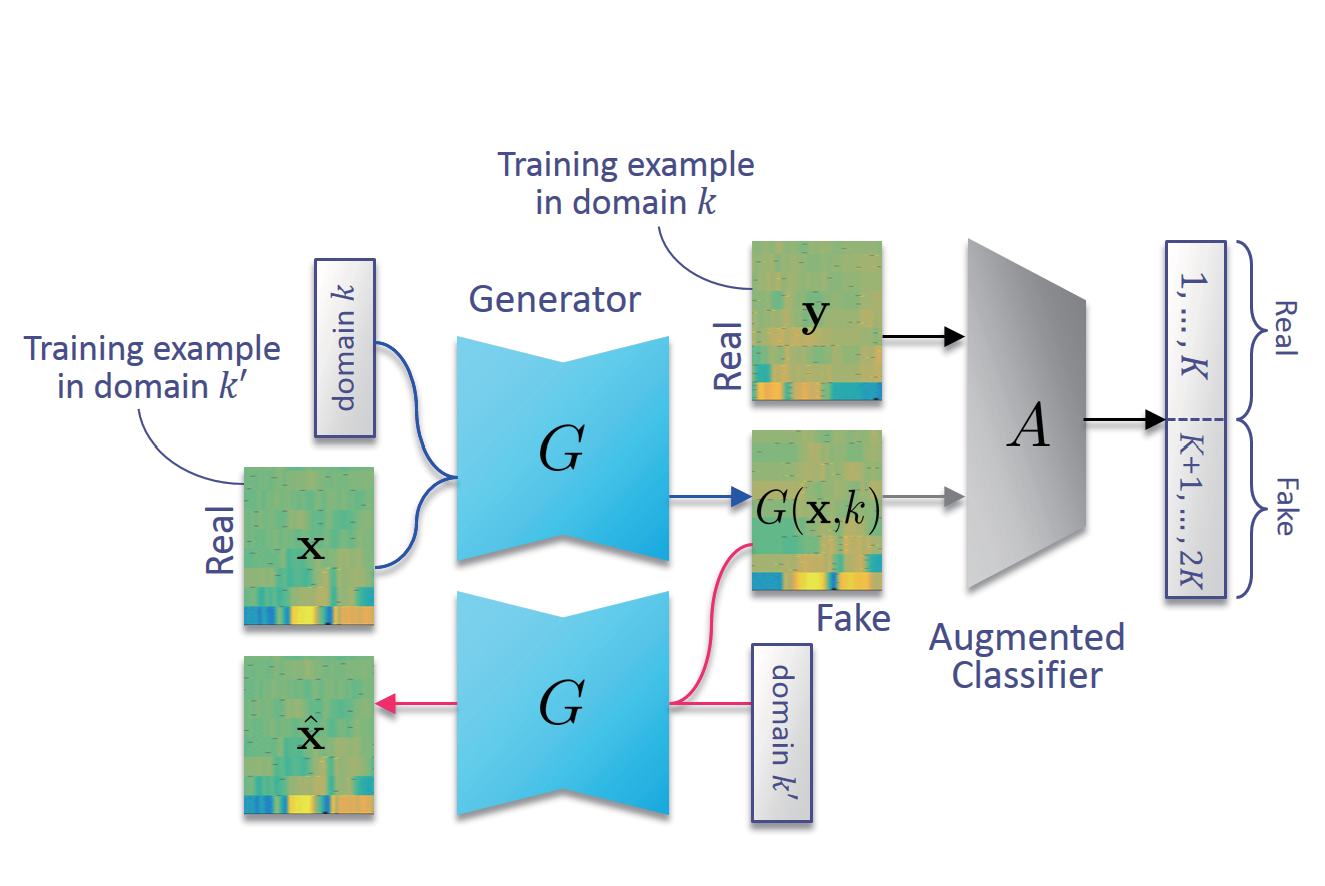}}
  \vspace{-1ex}
  \caption{Illustration of A-StarGAN training. The $A$ network is designed to
  produce $2K$ probabilities, where the first and second $K$ probabilities
  correspond to real and fake classes, and   
  simultaneously play the roles of the real/fake discriminator and domain classifier.}
  \label{fig:astargan}
  \end{minipage}
\end{figure*}

\subsection{Wasserstein StarGAN formulation}

Next, we describe the original StarGAN formulation
\cite{Choi2017short}.
It is frequently reported that optimization in regular GAN training 
can often get unstable.
It has been shown
that
using a cross-entropy measure as the minimax objective
corresponds to optimizing  
the Jensen-Shannon (JS) divergence between the real data distribution
and the generator's distribution \cite{Goodfellow2014short}.
As discussed in \cite{Arjovsky2017short}, 
why regular GAN training tends to easily get unstable 
can be explained by the fact that the JS divergence will be maxed out 
when the two distributions are distant from each other and thus have disjoint supports.
It is probable that this can also happen in the StarGAN training when
using a cross-entropy measure.
With the aim of stabilizing training, instead of the cross-entropy measure,
the original StarGAN adopts 
the Wasserstein distance as the training objective, 
which provides a meaningful distance metric
between two distributions even for those with disjoint supports. 
By using
the Kantorovich-Rubinstein duality theorem \cite{villani2008optimal},
a tractable form of 
the Wasserstein distance
between 
the real speech feature distribution $p(\Vec{x})$ and
the distribution of the fake samples generated by the generator $G(\Vec{x},k)$, 
where $\Vec{x}\sim p_{\rm d}(\Vec{x})$ and $k\sim p(k)$,
is given by
\begin{multline}
\mathcal{W}(G) = 
\max_{D\in\mathcal{D}}
\big\{
\mathbb{E}_{\Vec{y}\sim p_{\rm d}(\Vec{y})}[D(\Vec{y})] 
\\-
\mathbb{E}_{k\sim p(k), \Vec{x}\sim p_{\rm d}(\Vec{x})}[D(G(\Vec{x},k))] 
\big\},
\label{eq:wasserstein}
\end{multline} 
where 
$D$ must lie within the space
$\mathcal{D}$ of 1-Lipschitz functions.
A 1-Lipschtiz function is a differentiable function that has gradients 
with norm at most 1 everywhere. 
This Lipschtiz constraint is derived as a result of 
obtaining the above form of the Wasserstein distance \cite{villani2008optimal}.
As \refeq{wasserstein} shows, 
the computation of the Wasserstein distance requires
optimization with respect to a function $D$.
Thus, if we describe $D$ using a neural network,
the problem of minimizing $\mathcal{W}(G)$ with respect to $G$ 
leads to a minimax game played by $G$ and $D$, as with regular GAN training, 
where $D$ plays a similar role to the discriminator.
Now, recall that
the function $D$ must be 1-Lipschitz. 
Although there are several ways to 
constrain $D$, such as the weight clipping technique 
adopted in \cite{Arjovsky2017short},
one successful and convenient way involves imposing
a penalty on the sampled gradients of $D$
\begin{align}
\mathcal{R}(D) = 
\mathbb{E}_{\hat{\Vec{x}}\sim p(\hat{\Vec{x}})}
[(\|
\nabla
D(\hat{\Vec{x}}) \|_2 -1)^2],
\end{align}
and including it
in the training objective \cite{Gulrajani2017short},
where 
$\nabla$ denotes the gradient operator and
$\hat{\Vec{x}}$
is a sample uniformly drawn along a straight line between
a pair of a real and a generated samples. 
We must also consider incorporating the domain classification loss to 
encourage $G(\Vec{x},k)$ to belong to class $k$ and 
the cycle-consistency loss to 
encourage $G(\Vec{x},k)$ to preserve 
the linguistic information in the input $\Vec{x}$. 
Overall, the training objectives to be minimized with respect to $G$, $D$, and $C$
become
\begin{align}
\mathcal{I}_{G}(G) =& 
-
\lambda_{\rm adv}
\mathbb{E}_{k\sim p(k), \Vec{x}\sim p_{\rm d}(\Vec{x})}[D(G(\Vec{x},k))]
\nonumber\\
&
+
\lambda_{\rm cls}
\mathcal{L}_{\rm cls}^{G}(G)
+
\lambda_{\rm cyc}
\mathcal{L}_{\rm cyc}(G)
+
\lambda_{\rm id}
\mathcal{L}_{\rm id}(G),
\\
\mathcal{I}_{D}(D) =&
\lambda_{\rm adv}
\mathbb{E}_{k\sim p(k), \Vec{x}\sim p_{\rm d}(\Vec{x})}[D(G(\Vec{x},k))]
\nonumber\\
&- \lambda_{\rm adv}
\mathbb{E}_{\Vec{y}\sim p_{\rm d}(\Vec{y})}[D(\Vec{y})]
\nonumber\\
&+ \lambda_{\rm gp} \mathbb{E}_{\hat{\Vec{x}}\sim p(\hat{\Vec{x}})}
[(\|
\nabla
D(\hat{\Vec{x}}) \|_2 -1)^2],
\\
\mathcal{I}_{C}(C) =& 
\lambda_{\rm cls}
\mathcal{L}_{\rm cls}^{C}(C),
\end{align}
where $\lambda_{\rm gp}\ge 0$
is for weighing the importance of the gradient penalty.
We refer to this version of StarGAN as ``Wasserstein StarGAN (W-StarGAN)''.
It should be noted that 
the authors of
\cite{Choi2017short}
choose to implement $D$ and $C$ as a single
multi-task classifier network that simultaneously 
produces the values $D(\Vec{x})$ and $p_{C}(k|\Vec{x})$ $(k=1,\ldots,K)$ (\reffig{wstargan}). 

\subsection{Proposed New StarGAN formulation}
\label{subsec:astargan}
With the two StarGAN formulations presented above, 
the ability of $G$ to appropriately convert its input into a target domain
depends on how the decision boundary is formed by $C$ during training.
The domain classification loss can be easily made almost 0 
by letting 
the samples of $G(\Vec{x},k)$ resemble, for example, only 
a few of the real speech samples in domain $k$
near the decision boundary.
In such situations, 
$G$ will have no incentive to attempt to make the generated samples
get closer to the rest of the real speech samples distributed in domain $k$.
As a result, the conversion effect of the trained $G$ will be limited. 
One reasonable way 
to avoid such situations
would be to 
consider additional classes for out-of-distribution samples 
that do not belong to any of the domains
and encourage $G$ to not generate samples belonging to those classes. 
This idea can be formulated as follows.

First, we unify the real/fake discriminator and the domain classifier
into a single multiclass classifier $A$ that  
outputs $2K$ probabilities $p_A(k|\Vec{x})$ $(k=1,\ldots,2K)$
where $k=1,\ldots,K$ and $k=K+1,\ldots,2K$
correspond to the real domain classes and the fake classes, respectively.
Note that this differs from the multi-task classifier network 
mentioned above in that 
$p_A(k|\Vec{x})$ must now satisfy 
$\sum_{k=1}^{2K} p_A(k|\Vec{x})=1$.
Here, the $K$ fake classes can be seen as the classes for out-of-distribution samples.
Next, by using this multiclass classifier, 
we define 
\begin{align}
\mathcal{L}^A_{\rm adv}(A) =& 
-\mathbb{E}_{k\sim p(k), \Vec{y}\sim p_{\rm d}(\Vec{y}|k)}[\log p_A(k|\Vec{y})]
\nonumber\\
&-\mathbb{E}_{k\sim p(k), \Vec{x}\sim p_{\rm d}(\Vec{x})}[\log p_A(K\!+\!k|G(\Vec{x},k))],
\label{eq:advloss3_c}\\
\mathcal{L}^G_{\rm adv}(G) =&
- \mathbb{E}_{k\sim p(k), \Vec{x}\sim p_{\rm d}(\Vec{x})}[\log p_A(k|G(\Vec{x},k))]
\nonumber\\
&+ \mathbb{E}_{k\sim p(k), \Vec{x}\sim p_{\rm d}(\Vec{x})}[\log p_A(K\!+\!k|G(\Vec{x},k))],
\label{eq:advloss3_g}
\end{align}
as adversarial losses for classifier $A$ and generator $G$.
$\mathcal{L}^A_{\rm adv}(A)$ becomes small when $A$ correctly classifies 
$\Vec{y}\sim p_{\rm d}(\Vec{y}|k)$ as real speech samples in domain $k$
and $G(\Vec{x},k)$ as fake samples in domain $k$, whereas
$\mathcal{L}^G_{\rm adv}(G)$ becomes 
small when $G$ fools $A$ so that $G(\Vec{x},k)$ is misclassified by $A$ as 
real speech samples in domain $k$ and is not classified as fake samples.

We will show below that this minimax game reaches a global optimum when 
$p_{\rm d}(\Vec{y}|k) = p_{G}(\Vec{y}|k)$ for $k=1,\ldots,K$ 
if both $G$ and $A$ have infinite capacity, 
where $p_G(\Vec{y}|k)$ denotes 
the distribution of $\Vec{y} = G(\Vec{x},k)$ with $\Vec{x} \sim p_{\rm d}(\Vec{x})$.
We first consider the optimal classifier $A$ for any given generator $G$.
\begin{proposition}
For fixed $G$, 
$\mathcal{L}^A_{\rm adv}(A)$ is minimized when
\begin{align}
p_{A}^*(k|\Vec{y}) &= 
\frac{
p(k) p(\Vec{y}|k)
}{
\sum_k
p(k) p_{\rm d}(\Vec{y}|k)
+
\sum_k
p(k) p_{G}(\Vec{y}|k)
},
\label{eq:p_C_opt_1toK}
\\
p_{A}^*(K\!+\!k|\Vec{y})
&=
\frac{
p(k) p_{G}(\Vec{y}|k)
}{
\sum_k
p(k) p_{\rm d}(\Vec{y}|k)
+
\sum_k
p(k) p_{G}(\Vec{y}|k)
},
\label{eq:p_C_opt_K+1to2K}
\end{align}
for $k=1,\ldots,K$.
\end{proposition}
\begin{proof} 
By differentiating the Lagrangian
\begin{align}
{L}(A,\gamma) = 
\mathcal{L}^A_{\rm adv}(A) 
+ 
\int \gamma(\Vec{y})
\Bigg(
\sum_{k=1}^{2K} p_{A}(k|\Vec{y}) - 1 
\Bigg) 
\mbox{d}\Vec{y}
\end{align}
with respect to $p_{A}(k|\Vec{y})$
\begin{align}
\frac{\partial {L}(A,\gamma)}{\partial p_{A}(k|\Vec{y})}
=
\begin{cases}
- 
\frac{
p(k) p_{\rm d}(\Vec{y}|k)
}{
p_{A}(k|\Vec{y})
} + \gamma(\Vec{y})
&(1\le k \le K)
\\
-
\frac{
p(k) p_{G}(\Vec{y}|k)
}{
p_{A}(k|\Vec{y})
} + \gamma(\Vec{y})
&(K\!+\!1\le k \le 2K)
\end{cases}
\nonumber
\end{align}
and setting the result at zero, 
we obtain
\begin{align}
p_A(k|\Vec{y}) = 
\begin{cases}
p(k)p_{\rm d}(\Vec{y}|k)/\gamma(\Vec{y})
&(1\le k\le K)
\\
p(k)p_{G}(\Vec{y}|k)/\gamma(\Vec{y})
&(K\!+\!1\le k\le 2K)
\end{cases}.
\label{eq:p_C_opt_1}
\end{align}
Since $p_{A}(k|\Vec{y})$ must sum to unity,
the multiplier $\gamma$ must be
\begin{align}
\gamma(\Vec{y}) = \sum_{k=1}^K p(k)p_{\rm d}(\Vec{y}|k) + \sum_{k=1}^{K} p(k)p_{G}(\Vec{y}|k).
\label{eq:p_C_opt_2}
\end{align}
Substituting \refeq{p_C_opt_2} into \refeq{p_C_opt_1} concludes the proof.
\end{proof}

\begin{theorem}
The global optimum of the minimax game is achieved when 
$p_{\rm d}(\Vec{y}|k) = p_{G}(\Vec{y}|k)$ for $k=1,\ldots,K$.
\end{theorem}
\begin{proof}
By substituting \refeqs{p_C_opt_1toK}{p_C_opt_K+1to2K} into 
$\mathcal{L}^G_{\rm adv}(G)$, we can describe 
it as a function of $G$ only:
\begin{align}
\mathcal{L}^G_{\rm adv}(G) 
&=
- \mathbb{E}_{k\sim p(k), \Vec{y}\sim p_{G}(\Vec{y}|k)}
\left[\log 
\frac{
p_A^*(k|\Vec{y})
}{
p_A^*(K+k|\Vec{y})
}
\right]
\nonumber\\
&=
\mathbb{E}_{k\sim p(k), \Vec{y}\sim p_{G}(\Vec{y}|k)}
\left[\log 
\frac{p_{G}(\Vec{y}|k)}{p_{\rm d}(\Vec{y}|k)}
\right]
\nonumber\\
&=
\mathbb{E}_{k\sim p(k)}
{\rm KL}
[
p_{G}(\Vec{y}|k)\|
p_{\rm d}(\Vec{y}|k)
]
,
\end{align}
where ${\rm KL}[\cdot\|\cdot]$ denotes the
Kullback-Leibler (KL) divergence. 
Obviously, $\mathcal{L}^G_{\rm adv}(G)$ becomes 0 if and only if 
$p_{\rm d}(\Vec{y}|k) = p_{G}(\Vec{y}|k)$ for $k=1,\ldots,K$,
thus concluding the proof.
\end{proof}

As with 
the first two formulations,
we must also consider incorporating 
the cycle-consistency and identity mapping
losses to 
encourage $G(\Vec{x},k)$ to preserve 
the linguistic information in the input $\Vec{x}$. 
Overall, the training objectives to be minimized with respect to $G$ and $A$
become
\begin{align}
\mathcal{I}_{G}(G) =&
\lambda_{\rm adv}
\mathcal{L}^G_{\rm adv}(G) 
+
\lambda_{\rm cyc}
\mathcal{L}_{\rm cyc}(G)
+
\lambda_{\rm id}
\mathcal{L}_{\rm id}(G),
\\
\mathcal{I}_{A}(A) =&
\lambda_{\rm adv} 
\mathcal{L}^A_{\rm adv}(A).
\end{align}
We refer to this formulation as the ``augmented classifier StarGAN (A-StarGAN)'' (\reffig{astargan}).

A comparative look at the C-StarGAN \cite{Kameoka2018SLTshort_StarGAN-VC}
and A-StarGAN formulations  
may provide intuitive insights into the behavior of the A-StarGAN training.
Although not explicitly stated, 
with C-StarGAN, 
the minimax game played by $G$ and $D$
using \refeqs{advloss_d}{advloss_g} alone is shown to
correspond to 
minimizing the JS divergence between $p_G(\Vec{y}|k)$ and $p_{\rm d}(\Vec{y}|k)$.
While this minimax game 
only cares whether $G(\Vec{x},k)$ resembles real samples in domain $k$  
and is not concerned with whether $G(\Vec{x},k)$ is likely to belong to a different domain $k' \neq k$,
A-StarGAN
is designed to require $G(\Vec{x},k)$ to keep away from all the speaker domains except $k$ 
by explicitly penalizing 
$G(\Vec{x},k)$ for resembling real samples in domain $k' \neq k$.
We expect that this particular mechanism can contribute to enhancing the conversion effect.  
The domain classification loss given as \refeq{clsloss_g}
in C-StarGAN is expected to play this role; 
however, its effect can be limited for the reason already mentioned.
With A-StarGAN, the classifier augmented with 
the fake classes creates
additional decision boundaries, each of which is expected to 
partition the region of each domain 
into in-distribution and out-of-distribution regions
thanks to the adversarial learning and thus 
encourage the generator to generate samples that resemble real in-distribution samples only.
It should also be noted that in C-StarGAN,
when the domain classification loss comes into play, 
the training objective does not allow for 
an interpretation of the optimization process as distribution fitting, unlike in A-StarGAN. 
This is also true for the W-StarGAN formulation.

From the above discussion,
we can also think of another version of the A-StarGAN formulation,
in which the $K$ fake classes are merged into a single fake class 
(so the classifier $A$ now produces only $K+1$ probabilities) and 
the adversarial losses for classifier $A$ and generator $G$ are defined as
\begin{align}
\mathcal{L}^A_{\rm adv}(A) =& 
-\mathbb{E}_{k\sim p(k), \Vec{y}\sim p_{\rm d}(\Vec{y}|k)}[\log p_A(k|\Vec{y})]
\nonumber\\
&-\mathbb{E}_{k\sim p(k), \Vec{x}\sim p_{\rm d}(\Vec{x})}[\log p_A(K\!+\!1|G(\Vec{x},k))],
\label{eq:advloss4_c}\\
\mathcal{L}^G_{\rm adv}(G) =&
- \mathbb{E}_{k\sim p(k), \Vec{x}\sim p_{\rm d}(\Vec{x})}[\log p_A(k|G(\Vec{x},k))]
\nonumber\\
&+ \mathbb{E}_{k\sim p(k), \Vec{x}\sim p_{\rm d}(\Vec{x})}[\log p_A(K\!+\!1|G(\Vec{x},k))].
\label{eq:advloss4_g}
\end{align}
It should be noted that
the minimax game using these losses no longer leads to
the minimization of the KL divergence between $p_G(\Vec{y}|k)$ and $p_{\rm d}(\Vec{y}|k)$.
However, we still believe it can work reasonably well
if the augmented classifier really behaves in the way discussed above. 

\subsection{Acoustic feature}

In this paper, we choose to use mel-cepstral coefficients (MCCs) computed 
from a spectral envelope obtained using WORLD \cite{Morise2016short,pyworld_url}
as the acoustic feature to be converted.
Although it would also be interesting to consider 
directly converting time-domain signals (for example, like in \cite{Serra2019short}), 
given the recent significant advances in high-quality neural vocoder systems \cite{vandenOord2016short,Tamamori2017short,Kalchbrenner2018short,Mehri2016short,Jin2018short,vandenOord2017short,Ping2019short,Prenger2018short,Kim2018short,Wang2018short,Tanaka2018short}, 
we would expect to be able to generate high-quality signals 
using a neural vocoder once we obtain a sufficient set of acoustic features.  
Such systems can be advantageous in that
the model size for the generator can be made small enough to 
allow the system to run in real-time and 
work well even when a limited amount of training data is available.

At training time, 
we normalize 
each element $x_{q,n}$ of 
the MCC sequence $\Vec{x}$ to
$x_{q,n} \leftarrow (x_{q,n} - \psi_q)/\zeta_q$
where $q$ denotes the dimension index of the MCC sequence,
$n$ denotes the frame index, 
and $\psi_q$ and $\zeta_q$ denote
the means and standard deviations of the $q$-th MCC sequence 
within all the \footnote{
We chose to compute the mean and standard deviation only from the voiced segments in the training samples since we wanted them to be less dependent on the lengths of the silent segments. 
The voiced segments were detected using WORLD.
}{voiced} segments of the training samples of the same speaker.

\subsection{Conversion process}
\label{subsec:conversion}

After training $G$, we can convert the acoustic feature sequence $\Vec{x}$ 
of an input utterance with
\begin{align}
\hat{\Vec{y}} = G(\Vec{x}, k),
\end{align}
where $k$ denotes the target domain.
Once $\hat{\Vec{y}}$
has been obtained, 
we adjust
the mean and variance of the generated feature sequence 
so that they match the pretrained 
mean and variance of the feature vectors of the target speaker.
We can then generate a time-domain signal using the WORLD vocoder or any 
recently developed neural vocoder \cite{vandenOord2016short,Tamamori2017short,Kalchbrenner2018short,Mehri2016short,Jin2018short, vandenOord2017short, Ping2019short,Prenger2018short,Kim2018short,Wang2018short,Tanaka2018short,Kumar2019,Yamamoto2020}.

\subsection{Network architectures}
\label{subsec:netarch}

The architectures of all the networks are 
detailed in \reffigss{generator2d}{classifier}.
As detailed below, 
$G$ is designed to take an acoustic feature sequence as an input
and output an acoustic feature sequence of the same length  
so as to learn conversion rules that capture time dependencies.
Similarly, 
$D$, $C$ and $A$ are designed to take acoustic feature sequences as inputs
and generate sequences of probabilities.
There are two ways to incorporate the class index $k$ into $G$ or $D$.
One is to simply represent it as a one-hot vector 
and append it to the input of each layer.
The other is to retrieve a continuous vector 
given $k$ from a dictionary of embeddings and append it to each layer input,
as in our previous work \cite{Kameoka2020_ConvS2S-VC,Kameoka2020_VTN}.
In this work, we adopted the former way though both performed almost the same.
As detailed in \reffigss{generator2d}{classifier}, all the networks are designed
using fully convolutional architectures with gated linear units
(GLUs) \cite{Dauphin2017short}. The output of the GLU
block used is defined as 
$\mathsf{GLU}(\Vec{X}) = \Vec{X}_1 \odot \mathsf{sigmoid}(\Vec{X}_2)$, 
where $\Vec{X}$ is the layer input, 
$\Vec{X}_1$ and $\Vec{X}_2$ are equally sized arrays into which $\Vec{X}$ 
is split along the channel dimension,
and $\mathsf{sigmoid}$ is a sigmoid gate function.
Like long short-term memory units, 
GLUs can reduce the vanishing gradient
problem for deep architectures by providing a linear path for
the gradients while retaining nonlinear capabilities.

\noindent
{\bf Generator:}
As described in \reffigs{generator2d}{generator1d},
we use 
 a 2D CNN 
or a 1D CNN
that takes 
an acoustic feature sequence $\Vec{x}$ as an input
to design $G$, where $\Vec{x}$ is treated
as an image of size $Q\times N$ with $1$ channel in the 2D case 
or as a signal sequence of length $N$ with $Q$ channels in the 1D case.

\noindent
{\bf Real/Fake Discriminator:}
We leverage the idea of PatchGANs \cite{Isola2017short} to
design a real/fake discriminator or a Lipschitz continuous function $D$, 
which assigns a probability or a score 
to each local segment of an input feature sequence to
indicate whether it is real or fake. 
More specifically, $D$ takes 
an acoustic feature sequence $\Vec{y}$ as an input and 
produces a sequence of probabilities (with C-StarGAN)
or scores (with W-StarGAN) that measures how likely each segment of $\Vec{y}$ 
is to be real speech features.
With C-StarGAN, the final output of $D$ is given by the product of all these probabilities,
and with W-StarGAN, the final output of $D$ is given by the sum of all these scores.

\noindent
{\bf Domain Classifier/Augmented Classifier:}
We also design the domain classifier $C$ and the augmented classifier $A$ 
so that each of them takes 
an acoustic feature sequence $\Vec{y}$ as an input and 
produces a sequence of class probability distributions that measure 
how likely each segment of $\Vec{y}$ is to belong to domain $k$.
The final output of $p_C(k|\Vec{y})$ or $p_A(k|\Vec{y})$
is given by the product of all these distributions.

\begin{figure*}[t!]
\centering
  \centerline{\includegraphics[height=3.3cm]{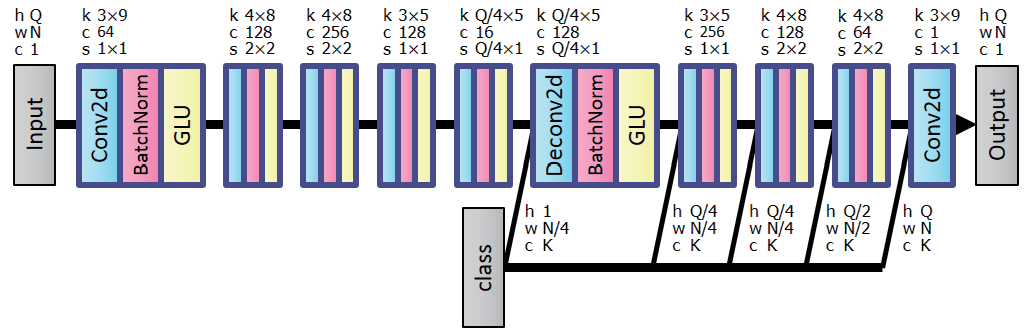}}
  \vspace{-0ex}
  \caption{Network architectures of the generator designed using 2D convolution layers. Here, the input and output of each layer are interpreted as images, where ``h'', ``w'' and ``c'' denote the height, width, and channel number, respectively. ``Conv2d'', ``BatchNorm'', ``GLU'', and ``Deconv2d'' denote 2D convolution, batch normalization, gated linear unit, and 2D transposed convolution layers, respectively. 
Batch normalization is applied to each channel and each height of the input.
``k'', ``c'' and ``s'' denote the kernel size, output channel number, and stride size of a convolution layer, respectively.  
The class index, represented as a one-hot vector, is concatenated
to the input of each convolution layer along the channel direction
after being repeated along the height and width directions 
so that it has a shape compatible with the input. 
}
\label{fig:generator2d}
\medskip
\centering
  \centerline{\includegraphics[height=3.3cm]{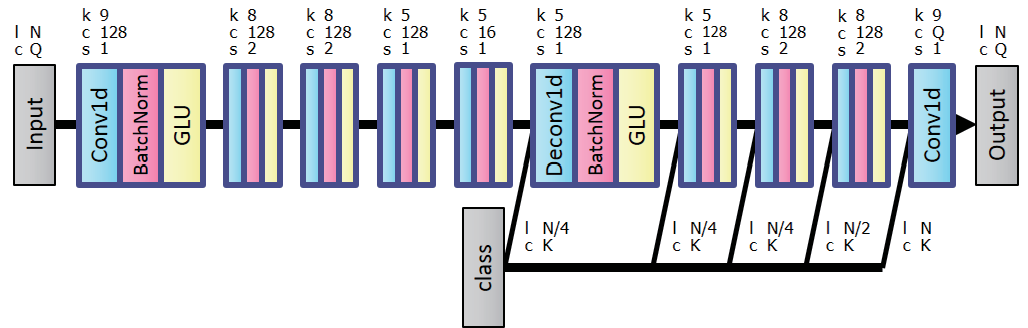}}
  \vspace{-0ex}
  \caption{Network architectures of the generator designed using 1D convolution layers. 
  Here, the input and output of the generator are interpreted as signal sequences, where ``l'' and ``c'' denote the length and channel number, respectively. ``Conv1d'', ``BatchNorm'', ``GLU'', and ``Deconv1d'' denote 1D convolution, batch normalization, gated linear unit, and 1D transposed convolution layers, respectively.  
Batch normalization is applied to each channel of the input.
The class index vector is concatenated
to the input of each convolution layer after being repeated along the time direction.}
\label{fig:generator1d}
\medskip
\centering
\begin{minipage}[t]{.48\linewidth}
  \centerline{\includegraphics[height=3.3cm]{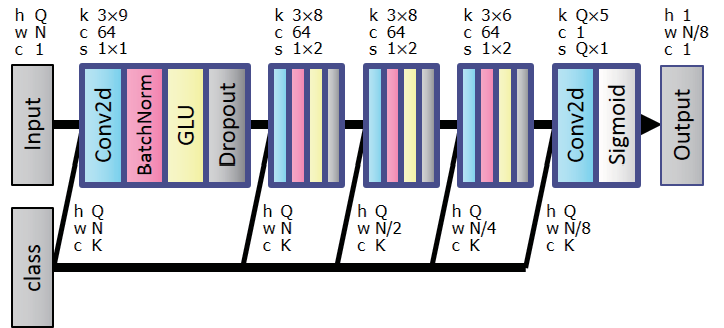}}
  \vspace{-0ex}
  \caption{Network architectures of the conditional discriminator in C-StarGAN. ``Sigmoid'' denotes an element-wise sigmoid function.}
\label{fig:discriminator}
\end{minipage}
\hfill
\begin{minipage}[t]{.48\linewidth}
  \centerline{\includegraphics[height=3.3cm]{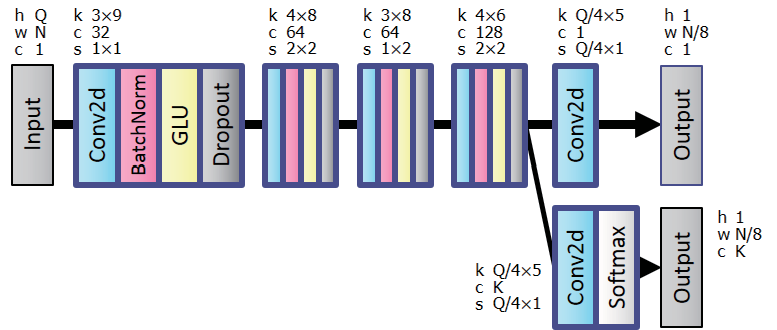}}
  \vspace{-0ex}
  \caption{Network architectures of the multi-task classifier in W-StarGAN.
  ``Softmax'' denotes a softmax function applied to the channel dimension.}
\label{fig:multi-task_discriminator}
\end{minipage}
\\
\medskip
  \centerline{\includegraphics[height=1.86cm]{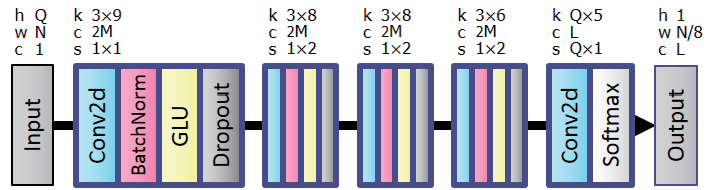}}
  \vspace{-0ex}
  \caption{Network architectures of the classifier in C-StarGAN and A-StarGAN.
  The output channel number $L$ is set to $K$ for the domain classifier in C-StarGAN, $2K$ for the augmented classifier in A-StarGAN1, and $K+1$ for the augmented classifier in A-StarGAN2.
  The channel number $M$ in the intermediate layers is set to 16 for the domain classifier in C-StarGAN and 64 for the augmented classifier in A-StarGAN1 and A-StarGAN2, respectively.}
\label{fig:classifier}
\end{figure*}

\section{Experiments}
\label{sec:experiments}

\subsection{Datasets}
\label{subsec:datasets}

To confirm the effects of the proposed StarGAN formulations, 
we conducted objective and subjective evaluation experiments involving 
a nonparallel speaker identity conversion task.
For the experiments, we used two datasets,
the CMU ARCTIC database \cite{Kominek2004short} and 
the Voice Conversion Challenge (VCC) 2018 dataset \cite{Lorenzo-Trueba2018short}.
The former consists of recordings of 
two female US English speakers (`clb' and `slt') and
two male US English speakers (`bdl' and `rms')
sampled at 16,000 Hz. 
The latter consists of recordings 
of six female and six male US English speakers
sampled at 22,050 Hz.
From the VCC2018 dataset, 
we selected two female speakers (`SF1' and `SF2') 
and two male speakers (`SM1' and `SM2').
Thus, for each dataset,
there were $K=4$ speakers, so 
in total there were twelve different combinations of source and target speakers.

\subsubsection{The CMU ARCTIC Dataset}

The CMU ARCTIC dataset consisted of four speakers, 
each reading the same 1,132 short sentences.
For each speaker, 
we used the first 1,000 and the latter 132 sentences
for training and evaluation.
To simulate a nonparallel training scenario,
we divided the first 1,000 sentences equally
into four groups and used only
the first, second, third, and fourth groups 
for speakers clb, bdl, slt, and rms,
so as not to use 
the same sentences between different speakers.
The training utterances of speakers clb, bdl, slt, and rms
were about 12, 11, 11, and 14 minutes long in total, respectively.
For each utterance, we extracted
a spectral envelope,
a logarithmic fundamental frequency (log $F_0$), and
aperiodicities (APs) every 8 ms using the
WORLD analyzer \cite{Morise2016short,pyworld_url}. 
We then extracted
$Q=28$ MCCs from 
each spectral envelope using the Speech Processing Toolkit (SPTK) \cite{pysptk_url}. 

\subsubsection{The VCC2018 Dataset}

The subset of the VCC2018 dataset consisted of four speakers,
each reading 
the same 116 short sentences (about seven minutes long in total).
For each speaker, 
we used the first 81 and the latter 35 sentences 
(about five and two minutes long in total) 
for training and evaluation.
Although we could actually construct a parallel corpus using this dataset,
we took care not to take advantage of it, because our purpose was 
to simulate a nonparallel training scenario.
For each utterance, 
we extracted
a spectral envelope,
a log $F_0$, APs, and $Q=36$ MCCs every 5 ms using the
WORLD analyzer \cite{Morise2016short,pyworld_url} and 
the SPTK \cite{pysptk_url} in the same manner. 

For both datasets,
the $F_0$ contours were converted using the logarithm Gaussian normalized
transformation described in \cite{Liu2007short}. 
The APs were used directly without modification. 
The signals of the converted speech were obtained using the methods described in \refsubsec{conversion}.

\subsection{Baseline Methods}

We chose the VAE-based \cite{Hsu2016short} and VAEGAN-based \cite{Hsu2017short} 
nonparallel VC methods
and our previously proposed CycleGAN-VC \cite{Kaneko2017dshort}
for comparison. 
In CycleGAN-VC,
we used the same network architectures shown in \reffigss{generator2d}{discriminator}
to design the generator and discriminator.
To clarify how close the proposed method can get to the performance achieved 
by one of the best-performing {\it parallel} VC methods,
we also chose a GMM-based open-source method 
called ``sprocket'' \cite{Kobayashi2018short} for comparison.
This method was used as a baseline in the VCC2018 \cite{Lorenzo-Trueba2018short}.
Note that
since sprocket is a parallel VC method, 
we tested it only on the VCC2018 dataset.
To run these methods,
we used the source codes provided by the authors \cite{Hsu2016url,Hsu2017url,Kobayashi2018url}.

\subsection{Hyperparameter Settings}

In the following, we use the abbreviations A-StarGAN1 and A-StarGAN2
to indicate the A-StarGAN formulations using \refeqs{advloss3_c}{advloss3_g} 
and using \refeqs{advloss4_c}{advloss4_g} 
as the adversarial losses. 
Hence, four different versions of the StarGAN formulations 
(namely C-StarGAN, W-StarGAN, A-StarGAN1 and A-StarGAN2) were considered for comparison.

All the networks were trained simultaneously with random initialization.
Adam optimization \cite{Kingma2015short} was used for model training, where 
the mini-batch size was 16. 
The settings of
the regularization parameters 
$\lambda_{\rm adv}$,
$\lambda_{\rm cls}$, $\lambda_{\rm cyc}$, $\lambda_{\rm id}$,
and $\lambda_{\rm gp}$,
the learning rates $\alpha_{G}$ and $\alpha_{D/C}$ for the generator and the discriminator/classifier, and the iteration number $I$
are listed in \reftab{hyperparameter_settings}.
For CycleGAN and
all the StarGAN versions,
the exponential decay rate
for the first moment was set at 0.9 for the generator and 
0.5 for the discriminator and classifier.
\reffig{loss_curves} shows the learning curves of 
C-StarGAN, W-StarGAN, and A-StarGAN1 under the above settings.
We also performed batch normalization with
the training mode at test time.
Note that all these hyperparameters were tuned on the VCC2018 database.

\begin{table}
\centering
\caption{Hyperparameter settings}
\begin{tabular}{r V{3} r|r|r|r}
\thline
&CycleGAN&C-StarGAN&W-StarGAN&A-StarGAN\\\thline
$\lambda_{\rm adv}$&1&1&10&1\\
$\lambda_{\rm cls}$&1&1&10&1\\
$\lambda_{\rm cyc}$&1&1&1&1\\
$\lambda_{\rm id}$&1&1&1&1\\
$\lambda_{\rm gp}$&--&--&10&--\\
$\alpha_{\rm G}$&$5\times 10^{-4}$&$5\times 10^{-4}$&$5\times 10^{-4}$&$5\times 10^{-4}$\\
$\alpha_{\rm D/C}$&$5\times 10^{-6}$&$2\times 10^{-6}$&$5\times 10^{-6}$&$2\times 10^{-6}$\\
$\rho$&1&1&1&1\\
$I$&$3.5\times 10^5$&$7\times 10^5$&$3.5\times 10^5$&$3.5\times 10^5$\\
\thline
\end{tabular}
\label{tab:hyperparameter_settings}
\end{table}

\begin{figure*}[t!]
\centering
\begin{minipage}[t]{.325\linewidth}
  \centerline{\includegraphics[width=.98\linewidth]{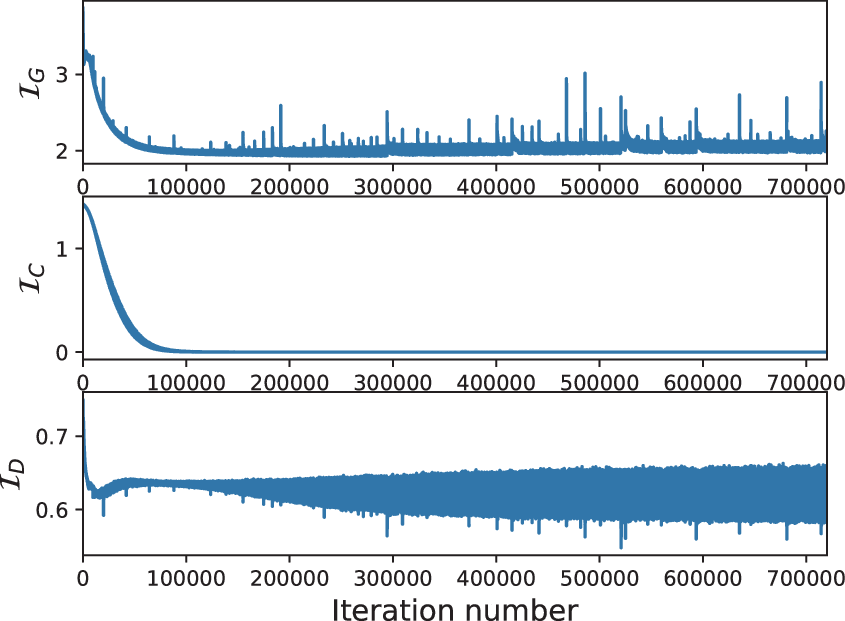}}
  \centerline{\scriptsize (a)}
\end{minipage}
\begin{minipage}[t]{.325\linewidth}
  \centerline{\includegraphics[width=.98\linewidth]{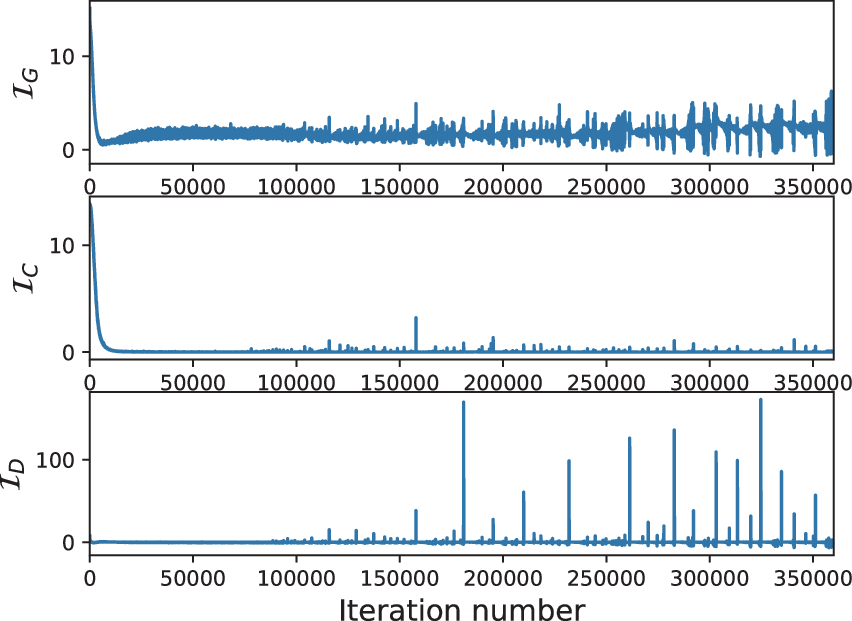}}
  \centerline{\scriptsize (b)}
\end{minipage}
\begin{minipage}[t]{.325\linewidth}
  \centerline{\includegraphics[width=.98\linewidth]{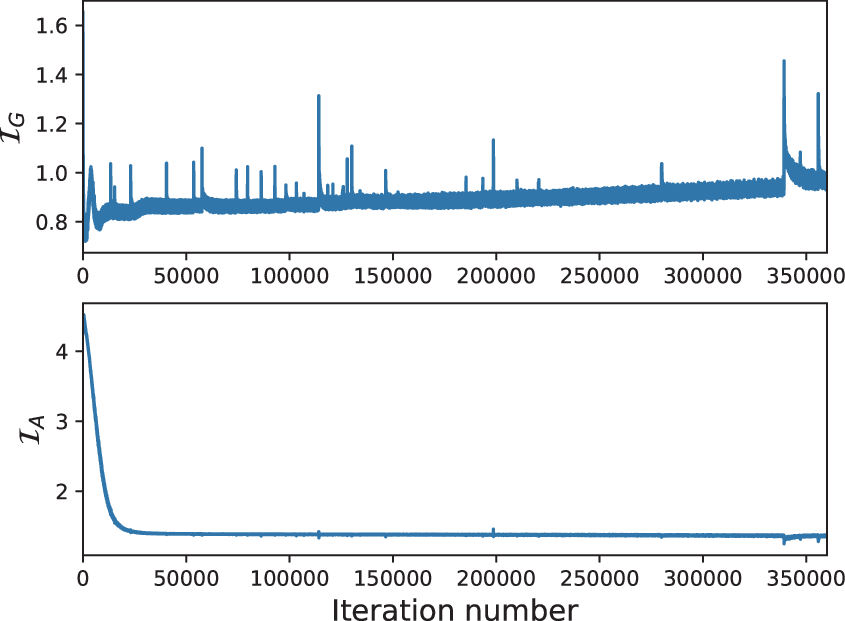}}
  \centerline{\scriptsize (c)}
\end{minipage}
\caption{Training loss curves of (a) C-StarGAN, (b) W-StarGAN, and (c) A-StarGAN1.}
\label{fig:loss_curves}
\end{figure*}

\subsection{Objective Performance Measure}

In each dataset,
the test set consists 
of speech samples of each speaker reading the same sentences.
Thus, 
the quality of a converted feature sequence can be assessed by
comparing it with the feature sequence of the target speaker reading the 
same sentence.
Here, we used the average of the mel-cepstral distortions (MCDs) 
taken along the 
dynamic time warping (DTW)
path between converted and target feature sequences
as the objective performance measure for each test utterance.

\subsection{Objective Evaluations}

\begin{table*}[t!]
\caption{MCD Comparisons of Different Network Configurations of $G$ on the CMU ARCTIC Dataset}
\centering
\begin{tabular}{l | l V{3} c|c|c|c|c|c|c|c|c|c}
\thline
\multicolumn{2}{c V{3}}{Speakers}&
\multicolumn{2}{c|}{CycleGAN}&
\multicolumn{2}{c|}{C-StarGAN}&
\multicolumn{2}{c|}{W-StarGAN}&
\multicolumn{2}{c|}{A-StarGAN1}&
\multicolumn{2}{c}{A-StarGAN2}
\\\hline
src&trg&1D&2D&1D&2D&1D&2D&1D&2D&1D&2D\\\thline
      &   bdl&$\!\!\bm{8.34\pm .16}\!\!$&$\!\!8.87\pm .15\!\!$
             &$\!\!\bm{7.84\pm .13}\!\!$&$\!\!8.47\pm .14\!\!$
             &$\!\!7.72\pm .13\!\!$&$\!\!\bm{7.37\pm .11}\!\!$
             &$\!\!\bm{7.50\pm .14}\!\!$&$\!\!8.04\pm .15\!\!$
             &$\!\!\bm{7.57\pm .14}\!\!$&$\!\!7.59\pm .12\!\!$\\
   clb&   slt&$\!\!7.13\pm .06\!\!$&$\!\!\bm{6.99\pm .06}\!\!$
             &$\!\!7.45\pm .07\!\!$&$\!\!\bm{6.87\pm .08}\!\!$
             &$\!\!7.02\pm .06\!\!$&$\!\!\bm{6.63\pm .05}\!\!$
             &$\!\!\bm{6.56\pm .06}\!\!$&$\!\!6.99\pm .07\!\!$
             &$\!\!\bm{6.64\pm .06}\!\!$&$\!\!6.92\pm .06\!\!$\\
      &   rms&$\!\!\bm{7.64\pm .06}\!\!$&$\!\!8.38\pm .08\!\!$
             &$\!\!8.31\pm .08\!\!$&$\!\!\bm{8.02\pm .08}\!\!$
             &$\!\!6.87\pm .06\!\!$&$\!\!\bm{6.81\pm .06}\!\!$
             &$\!\!\bm{7.01\pm .09}\!\!$&$\!\!7.39\pm .08\!\!$
             &$\!\!7.23\pm .06\!\!$&$\!\!\bm{7.22\pm .06}\!\!$\\\hline
      &   clb&$\!\!8.43\pm .14\!\!$&$\!\!\bm{8.41\pm .13}\!\!$
             &$\!\!8.03\pm .12\!\!$&$\!\!\bm{7.75\pm .13}\!\!$
             &$\!\!7.40\pm .12\!\!$&$\!\!\bm{7.03\pm .12}\!\!$
             &$\!\!\bm{7.57\pm .13}\!\!$&$\!\!7.80\pm .12\!\!$
             &$\!\!\bm{7.45\pm .14}\!\!$&$\!\!7.62\pm .12\!\!$\\
   bdl&   slt&$\!\!\bm{8.10\pm .11}\!\!$&$\!\!8.29\pm .13\!\!$
             &$\!\!8.24\pm .10\!\!$&$\!\!\bm{8.08\pm .12}\!\!$
             &$\!\!7.36\pm .10\!\!$&$\!\!\bm{6.85\pm .06}\!\!$
             &$\!\!\bm{7.06\pm .10}\!\!$&$\!\!7.66\pm .10\!\!$
             &$\!\!\bm{7.22\pm .09}\!\!$&$\!\!7.27\pm .08\!\!$\\
      &   rms&$\!\!8.15\pm .12\!\!$&$\!\!\bm{8.08\pm .14}\!\!$
             &$\!\!8.53\pm .12\!\!$&$\!\!\bm{8.09\pm .14}\!\!$
             &$\!\!7.63\pm .15\!\!$&$\!\!\bm{7.45\pm .17}\!\!$
             &$\!\!\bm{7.34\pm .16}\!\!$&$\!\!8.09\pm .12\!\!$
             &$\!\!\bm{7.45\pm .15}\!\!$&$\!\!7.72\pm .14\!\!$\\\hline
      &   clb&$\!\!7.18\pm .07\!\!$&$\!\!\bm{7.03\pm .07}\!\!$
             &$\!\!7.48\pm .09\!\!$&$\!\!\bm{7.04\pm .07}\!\!$
             &$\!\!7.10\pm .07\!\!$&$\!\!\bm{6.72\pm .08}\!\!$
             &$\!\!\bm{7.02\pm .08}\!\!$&$\!\!7.12\pm .09\!\!$
             &$\!\!\bm{6.75\pm .07}\!\!$&$\!\!7.19\pm .10\!\!$\\
   slt&   bdl&$\!\!\bm{8.48\pm .10}\!\!$&$\!\!8.68\pm .12\!\!$
             &$\!\!\bm{7.94\pm .08}\!\!$&$\!\!8.59\pm .12\!\!$
             &$\!\!7.79\pm .09\!\!$&$\!\!\bm{7.40\pm .08}\!\!$
             &$\!\!\bm{7.43\pm .09}\!\!$&$\!\!8.00\pm .10\!\!$
             &$\!\!\bm{7.53\pm .08}\!\!$&$\!\!7.60\pm .09\!\!$\\
      &   rms&$\!\!8.70\pm .10\!\!$&$\!\!\bm{8.52\pm .08}\!\!$
             &$\!\!8.67\pm .08\!\!$&$\!\!\bm{8.55\pm .10}\!\!$
             &$\!\!7.12\pm .10\!\!$&$\!\!\bm{7.00\pm .09}\!\!$
             &$\!\!\bm{7.26\pm .10}\!\!$&$\!\!7.78\pm .10\!\!$
             &$\!\!\bm{7.55\pm .10}\!\!$&$\!\!7.84\pm .13\!\!$\\\hline
      &   clb&$\!\!\bm{7.98\pm .07}\!\!$&$\!\!8.33\pm .10\!\!$
             &$\!\!8.04\pm .08\!\!$&$\!\!\bm{7.88\pm .06}\!\!$
             &$\!\!7.24\pm .08\!\!$&$\!\!\bm{6.94\pm .07}\!\!$
             &$\!\!\bm{7.23\pm .08}\!\!$&$\!\!7.48\pm .07\!\!$
             &$\!\!\bm{7.16\pm .07}\!\!$&$\!\!7.28\pm .07\!\!$\\
   rms&   bdl&$\!\!8.18\pm .17\!\!$&$\!\!\bm{8.12\pm .18}\!\!$
             &$\!\!\bm{8.04\pm .18}\!\!$&$\!\!8.49\pm .18\!\!$
             &$\!\!8.52\pm .19\!\!$&$\!\!\bm{8.16\pm .19}\!\!$
             &$\!\!\bm{7.60\pm .19}\!\!$&$\!\!8.33\pm .20\!\!$
             &$\!\!\bm{7.67\pm .15}\!\!$&$\!\!7.78\pm .19\!\!$\\
      &   slt&$\!\!8.99\pm .09\!\!$&$\!\!\bm{8.67\pm .09}\!\!$
             &$\!\!8.57\pm .11\!\!$&$\!\!\bm{8.52\pm .13}\!\!$
             &$\!\!7.41\pm .10\!\!$&$\!\!\bm{7.11\pm .10}\!\!$
             &$\!\!\bm{7.08\pm .09}\!\!$&$\!\!7.85\pm .11\!\!$
             &$\!\!\bm{7.34\pm .11}\!\!$&$\!\!7.55\pm .11\!\!$\\\hline
\multicolumn{2}{c V{3}}{All pairs}
             &$\!\!\bm{8.11\pm .04}\!\!$&$\!\!8.20\pm .04\!\!$
             &$\!\!8.09\pm .04\!\!$&$\!\!\bm{8.03\pm .04}\!\!$
             &$\!\!7.43\pm .04\!\!$&$\!\!\bm{7.12\pm .04}\!\!$
             &$\!\!\bm{7.22\pm .04}\!\!$&$\!\!7.71\pm .04\!\!$
             &$\!\!\bm{7.30\pm .03}\!\!$&$\!\!7.46\pm .03\!\!$\\\thline   
\end{tabular}
\label{tab:mcd_config_comp}
\end{table*}

\begin{table*}[t!]
\caption{MCD Comparisons with baseline methods on the CMU ARCTIC dataset}
\centering
\begin{tabular}{l | l V{3} c|c|c|c|c|c|c}
\thline
\multicolumn{2}{c V{3}}{Speakers}&
\multirow{2}{*}{VAE}&\multirow{2}{*}{VAEGAN}&
\multirow{2}{*}{CycleGAN}&\multirow{2}{*}{C-StarGAN}&
\multirow{2}{*}{W-StarGAN}&\multirow{2}{*}{A-StarGAN1}&
\multirow{2}{*}{A-StarGAN2}
\\\cline{1-2}
source&target&&&&&&&\\\thline
      &   bdl&$7.85\pm .10$&$8.82\pm .14$&$8.34\pm .16$&$8.47\pm .14$
             &$\bm{7.37\pm .11}$&$7.50\pm .14$&$7.57\pm .14$\\
   clb&   slt&$7.08\pm .07$&$8.11\pm .06$&$7.13\pm .06$&$6.87\pm .08$
             &$6.63\pm .05$&$\bm{6.56\pm .06}$&$6.64\pm .06$\\
      &   rms&$7.70\pm .05$&$8.14\pm .07$&$7.64\pm .06$&$8.02\pm .08$
             &$\bm{6.81\pm .06}$&$7.01\pm .09$&$7.23\pm .06$\\\hline
      &   clb&$7.68\pm .09$&$8.86\pm .10$&$8.43\pm .14$&$7.75\pm .13$
             &$\bm{7.03\pm .12}$&$7.57\pm .13$&$7.45\pm .14$\\
   bdl&   slt&$7.39\pm .09$&$8.15\pm .08$&$8.10\pm .11$&$8.08\pm .12$
             &$\bm{6.85\pm .06}$&$7.06\pm .10$&$7.22\pm .09$\\
      &   rms&$7.99\pm .12$&$8.28\pm .11$&$8.15\pm .12$&$8.09\pm .14$
             &$7.45\pm .17$&$\bm{7.34\pm .16}$&$7.45\pm .15$\\\hline
      &   clb&$6.96\pm .07$&$8.36\pm .09$&$7.18\pm .07$&$7.04\pm .07$
             &$\bm{6.72\pm .08}$&$7.02\pm .08$&$6.75\pm .07$\\
   slt&   bdl&$7.44\pm .08$&$7.60\pm .09$&$8.48\pm .10$&$8.59\pm .12$
             &$\bm{7.40\pm .08}$&$7.43\pm .09$&$7.53\pm .08$\\
      &   rms&$7.72\pm .10$&$8.39\pm .11$&$8.70\pm .10$&$8.55\pm .10$
             &$\bm{7.00\pm .09}$&$7.26\pm .10$&$7.55\pm .10$\\\hline
      &   clb&$7.81\pm .07$&$8.64\pm .09$&$7.98\pm .07$&$7.88\pm .06$
             &$\bm{6.94\pm .07}$&$7.23\pm .08$&$7.16\pm .07$\\
   rms&   bdl&$8.02\pm .15$&$8.19\pm .17$&$8.18\pm .17$&$8.49\pm .18$
             &$8.16\pm .19$&$\bm{7.60\pm .19}$&$7.67\pm .15$\\
      &   slt&$7.88\pm .09$&$8.20\pm .11$&$8.99\pm .09$&$8.52\pm .13$
             &$7.11\pm .10$&$\bm{7.08\pm .09}$&$7.34\pm .11$\\\hline
\multicolumn{2}{c V{3}}{All pairs}
             &$7.63\pm .03$&$8.31\pm .03$&$8.11\pm .04$&$8.03\pm .04$
             &$\bm{7.12\pm .04}$&$7.22\pm .04$&$7.30\pm .03$\\\thline   
\end{tabular}
\label{tab:mcd_baseline_comp_arctic}
\end{table*}

\begin{table*}[t!]
\caption{MCD Comparisons with baseline methods on the VCC2018 dataset}
\centering
\begin{tabular}{l | l V{3} c|c|c|c|c|c|c|c}
\thline
\multicolumn{2}{c V{3}}{Speakers}&\multicolumn{7}{c|}{nonparallel methods}&parallel method\\\hline
source&target&VAE&VAEGAN&CycleGAN&C-StarGAN&W-StarGAN&A-StarGAN1&A-StarGAN2&sprocket\\\thline
      &   SM1&$7.66\pm 0.12$&$7.70\pm 0.12$&$7.72\pm 0.13$&$7.52\pm 0.12$
             &$\bm{ 7.26\pm 0.12}$&$7.32\pm 0.13$&$7.27\pm 0.13$&$6.91\pm 0.12$\\
   SF1&   SF2&$7.53\pm 0.12$&$7.43\pm 0.12$&$7.35\pm 0.16$&$7.20\pm 0.14$
             &$7.16\pm 0.13$&$7.05\pm 0.12$&$\bm{6.98\pm 0.15}$&$6.70\pm 0.13$\\
      &   SM2&$8.06\pm 0.14$&$8.04\pm 0.15$&$7.91\pm 0.13$&$7.92\pm 0.14$
             &$7.67\pm 0.12$&$7.69\pm 0.12$&$\bm{7.58\pm 0.12}$&$7.06\pm 0.12$\\\hline
      &   SF1&$8.25\pm 0.10$&$8.20\pm 0.13$&$8.03\pm 0.12$&$7.87\pm 0.10$
             &$7.69\pm 0.10$&$7.58\pm 0.10$&$\bm{7.45\pm 0.10}$&$7.01\pm 0.11$\\
   SM1&   SF2&$7.43\pm 0.11$&$7.23\pm 0.12$&$6.95\pm 0.12$&$6.97\pm 0.12$
             &$6.95\pm 0.10$&$6.71\pm 0.12$&$\bm{6.66\pm 0.11}$&$6.30\pm 0.11$\\
      &   SM2&$7.92\pm 0.11$&$7.82\pm 0.10$&$7.20\pm 0.09$&$7.32\pm 0.11$
             &$7.24\pm 0.09$&$\bm{7.01\pm 0.11}$&$7.08\pm 0.10$&$6.58\pm 0.10$\\\hline
      &   SF1&$7.97\pm 0.13$&$7.83\pm 0.12$&$7.65\pm 0.13$&$7.59\pm 0.12$
             &$7.59\pm 0.10$&$7.43\pm 0.10$&$\bm{7.40\pm 0.11}$&$7.21\pm 0.11$\\
   SF2&   SM1&$7.38\pm 0.11$&$7.37\pm 0.10$&$7.04\pm 0.11$&$7.00\pm 0.11$
             &$6.91\pm 0.12$&$\bm{6.82\pm 0.12}$&$6.83\pm 0.13$&$6.77\pm 0.11$\\
      &   SM2&$7.92\pm 0.12$&$7.78\pm 0.11$&$7.64\pm 0.12$&$7.54\pm 0.13$
             &$7.45\pm 0.12$&$7.49\pm 0.13$&$\bm{7.48\pm 0.10}$&$6.85\pm 0.12$\\\hline
      &   SF1&$8.33\pm 0.15$&$8.20\pm 0.16$&$8.13\pm 0.17$&$8.01\pm 0.17$
             &$7.84\pm 0.15$&$7.75\pm 0.16$&$\bm{7.67\pm 0.14}$&$7.31\pm 0.12$\\
   SM2&   SM1&$7.73\pm 0.14$&$7.66\pm 0.14$&$7.20\pm 0.13$&$7.20\pm 0.12$
             &$7.07\pm 0.12$&$6.99\pm 0.13$&$\bm{6.97\pm 0.13}$&$6.88\pm 0.11$\\
      &   SF2&$7.74\pm 0.14$&$7.65\pm 0.14$&$7.34\pm 0.16$&$7.25\pm 0.15$
             &$7.27\pm 0.14$&$7.03\pm 0.15$&$\bm{6.98\pm 0.15}$&$6.78\pm 0.15$\\\hline
\multicolumn{2}{c V{3}}{All pairs}
             &$7.83\pm 0.05$&$7.74\pm 0.05$&$7.51\pm 0.05$&$7.45\pm 0.05$
             &$7.35\pm 0.04$&$7.24\pm 0.05$&$\bm{7.19\pm 0.05}$&$6.86\pm 0.04$\\\thline   
\end{tabular}
\label{tab:mcd_baseline_comp_vcc2018}
\end{table*}


First,
we evaluated the performance of each StarGAN version with
different network configurations of $G$. The detailed settings for
these configurations are shown in \reffigs{generator2d}{generator1d}.
The network architectures of the conditional discriminator and domain classifier
in C-StarGAN, the multi-task classifier in W-StarGAN, and the
augmented classifier in A-StarGAN are shown in 
\reffigss{discriminator}{classifier}.
\reftab{mcd_config_comp} shows the average MCDs 
along with standard errors
obtained with these network configurations. 
The results show 
that 
the CycleGAN, C-StarGAN, W-StarGAN, A-StarGAN1, and A-StarGAN2 methods
performed better 
with $G$ designed using 
1D-, 2D-, 2D-, 1D-, and 1D-CNNs, respectively.
In the following,
we only present the results obtained 
with these configurations.

\reftabs{mcd_baseline_comp_arctic}{mcd_baseline_comp_vcc2018}
show the MCDs obtained with the proposed and baseline methods.
As the results show, 
W-StarGAN 
and A-StarGAN1
performed best and next best on the CMU ARCTIC dataset,
and A-StarGAN2 and
A-StarGAN1 performed best and next best of all the nonparallel methods 
on the VCC 2018 dataset.
All the StarGAN versions
performed consistently better than CycleGAN.
Since both 
CycleGAN and C-StarGAN use the cross-entropy measure to define the adversarial losses, 
the superiority of C-StarGAN over 
CycleGAN reflects 
the effect of the many-to-many extension. 
Now, let us turn to the comparisons of the four StarGAN versions.
From the results, we can see that W-StarGAN performed better than C-StarGAN on both datasets, revealing the advantage of the training objective defined using the Wasserstein distance with the gradient penalty.
We also confirmed that A-StarGAN1\&2 
performed even better than W-StarGAN on the VCC2018 dataset,
though it performed
slightly worse on the CMU ARCTIC dataset. 
We also confirmed that all the StarGAN versions 
could not yield higher performance than sprocket.
Given the fact that sprocket had the
advantage of using parallel data for the model training, 
we consider the current result to be promising, since
the proposed methods are already advantageous in that they can
be applied in nonparallel training scenarios.

\begin{table*}[t!]
\begin{minipage}{.45\linewidth}
\caption{Real/Fake Discrimination Accuracy (\%)}
\centering
\begin{tabular}{l | l V{3} c|c|c}
\thline
\multicolumn{2}{c V{3}}{Speakers}&
\multirow{2}{*}{C-StarGAN}&
\multirow{2}{*}{A-StarGAN1}&
\multirow{2}{*}{A-StarGAN2}
\\\cline{1-2}
\!source\!\!&\!target\!&&&\\\thline
      &   bdl&$63.19\pm 2.79$&$48.10\pm .27$&$41.80\pm 1.93$\\
   clb&   slt&$76.56\pm 3.70$&$47.42\pm .41$&$28.43\pm 2.90$\\
      &   rms&$16.36\pm 4.97$&$51.51\pm .40$&$40.43\pm 2.72$\\\hline
      &   clb&$17.73\pm 2.43$&$46.77\pm .45$&$20.49\pm 2.10$\\
   bdl&   slt&$87.83\pm 1.98$&$47.64\pm .29$&$29.46\pm 2.38$\\
      &   rms&$3.93\pm 2.48$&$51.00\pm .37$&$38.14\pm 2.64$\\\hline
      &   clb&$22.56\pm 2.48$&$47.82\pm .41$&$22.37\pm 2.43$\\
   slt&   bdl&$69.25\pm 2.61$&$48.42\pm .30$&$40.71\pm 2.38$\\
      &   rms&$9.56\pm 3.78$&$50.53\pm .48$&$33.65\pm 2.71$\\\hline
      &   clb&$28.81\pm 2.31$&$48.48\pm .47$&$25.74\pm 2.32$\\
   rms&   bdl&$60.52\pm 2.78$&$48.60\pm .48$&$40.41\pm 2.41$\\
      &   slt&$72.18\pm 4.01$&$47.53\pm .55$&$29.19\pm 2.48$\\\hline
\multicolumn{2}{c V{3}}{All pairs}
             &$44.04\pm 1.73$&$48.65\pm .14$&$32.56\pm .80$\\\thline   
\end{tabular}
\label{tab:dis_output_arctic}
\end{minipage}
\begin{minipage}{.55\linewidth}
\caption{Speaker Classification Accuracy (\%)}
\centering
\begin{tabular}{l | l V{3} c|c|c|c}
\thline
\multicolumn{2}{c V{3}}{Speakers}&
\multirow{2}{*}{C-StarGAN}&
\multirow{2}{*}{W-StarGAN}&
\multirow{2}{*}{A-StarGAN1}&
\multirow{2}{*}{A-StarGAN2}
\\\cline{1-2}
\!source\!\!&\!target\!&&&&\\\thline
      &   bdl&$96.07\pm 2.92$&$99.99\pm .00$&$99.83\pm .05$&$98.58\pm .40$\\
   clb&   slt&$96.29\pm 2.67$&$99.70\pm .28$&$96.92\pm 1.98$&$80.19\pm 5.54$\\
      &   rms&$93.29\pm 3.72$&$99.97\pm .02$&$99.97\pm .01$&$92.34\pm 3.64$\\\hline
      &   clb&$94.23\pm 2.57$&$99.87\pm .07$&$99.38\pm .33$&$87.90\pm 3.53$\\
   bdl&   slt&$96.22\pm 2.27$&$98.63\pm .93$&$99.48\pm .23$&$86.98\pm 3.97$\\
      &   rms&$90.91\pm 3.95$&$99.98\pm .01$&$99.93\pm .02$&$93.43\pm 2.31$\\\hline
      &   clb&$80.78\pm 5.88$&$99.38\pm .49$&$98.78\pm .99$&$93.66\pm 1.96$\\
   slt&   bdl&$96.15\pm 2.58$&$99.99\pm .00$&$99.72\pm .07$&$93.94\pm 3.02$\\
      &   rms&$92.24\pm 3.98$&$99.99\pm .00$&$99.90\pm .05$&$90.95\pm 3.47$\\\hline
      &   clb&$86.04\pm 5.55$&$98.55\pm .94$&$99.42\pm .31$&$92.81\pm 2.76$\\
   rms&   bdl&$95.65\pm 2.54$&$99.99\pm .00$&$98.51\pm 1.14$&$95.74\pm 1.74$\\
      &   slt&$99.98\pm .01$&$99.99\pm .01$&$99.36\pm .42$&$89.11\pm 3.48$\\\hline
\multicolumn{2}{c V{3}}{All pairs}
             &$93.15\pm 1.06$&$99.67\pm .12$&$99.27\pm .22$&$91.30\pm .96$\\\thline   
\end{tabular}
\label{tab:cls_output_arctic}
\end{minipage}
\end{table*}

Balancing the learning of the players in a minimax game is essential in the GAN framework. 
The probabilities of the feature sequence converted from each test sample  
being real and produced by the target speaker 
may provide an indication of how successfully the generator, discriminator, and classifier
have been trained in a balanced manner.
\reftabs{dis_output_arctic}{cls_output_arctic} show
the mean outputs of the discriminator and classifier
of C-StarGAN, W-StarGAN, and A-StarGAN1\&2 at test time.
Note that since 
the discriminator in W-StarGAN produces scores (instead of probabilities), which are not straightforward to interpret, 
we 
have omitted them in \reftab{dis_output_arctic}.
As for the augmented classifier in A-StarGAN1\&2,
if we use $p_{k,n}$ to denote an element of the classifier output
corresponding to the probability of the classifier input 
belonging to class $k$ at \footnote{
Since the classifier is designed to contain 
three downsampling layers, the output sequence becomes $2^3=8$ times shorter than the input sequence. 
Hence, $n$ corresponds to a segment consisting of eight consecutive frames.
This is why here we have used ``segment'' rather than ``frame'' to signify $n$. 
}{segment $n$}, 
the values 
$\sum_{k=1}^{K} p_{k,n}$ 
and
$p_{k,n}/\sum_{k'=1}^{K} p_{k',n}$
correspond to the probabilities 
of the classifier input being real and produced by speaker $k$, 
respectively,
at that segment\footnote{
In A-StarGAN1, $k \in\{1,\ldots,K\}$ indicates a class corresponding to real speaker $k$,
whereas $K+k$ indicates a class corresponding to fake speaker $k$.
Thus,
the marginal probability of an input being real at segment $n$
can be expressed as $\sum_{k=1}^{K} p_{k,n}$.
The conditional probability of an input being produced by speaker $k$ at segment $n$, 
given that the input is real,
can be expressed as $p_{k,n}/\sum_{k'=1}^{K} p_{k',n}$.
The same applies to A-StarGAN2.
}{}. 
The means of these values over $n$ of all the utterances
along with standard errors
are shown in
\reftabs{dis_output_arctic}{cls_output_arctic}.
As \reftabs{dis_output_arctic}{cls_output_arctic} indicate, 
the generators in all the StarGAN versions
were successful in 
confusing the discriminator and
making the classifier believe
that the feature sequence converted from each test sample
was produced by the target speaker.

The modulation spectra of MCC sequences are known to be  
quantities that are closely related to the perceived quality and naturalness of speech \cite{Takamichi2016short}. 
\reffig{modspec_all} shows an example of the average 
modulation spectra of the converted MCC sequences obtained with the proposed and baseline methods
along with those of the real speech of the target speaker.
As this example shows, the modulation spectra obtained with 
the CycleGAN-based method and 
all the StarGAN-based methods were relatively closer to those of real speech 
than the spectra obtained with the VAE-based and VAEGAN-based methods over the entire frequency range, 
thanks to the adversarial training strategy. 

\begin{figure*}[t!]
\centering
\begin{minipage}[t]{.195\linewidth}
  \centerline{\includegraphics[width=.98\linewidth]{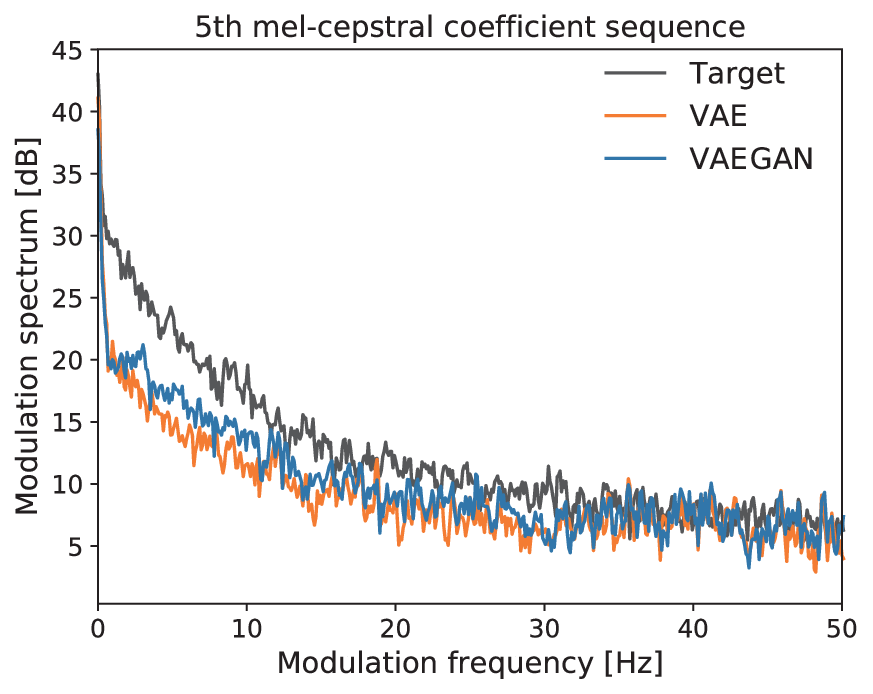}}
  \medskip
  \centerline{\includegraphics[width=.98\linewidth]{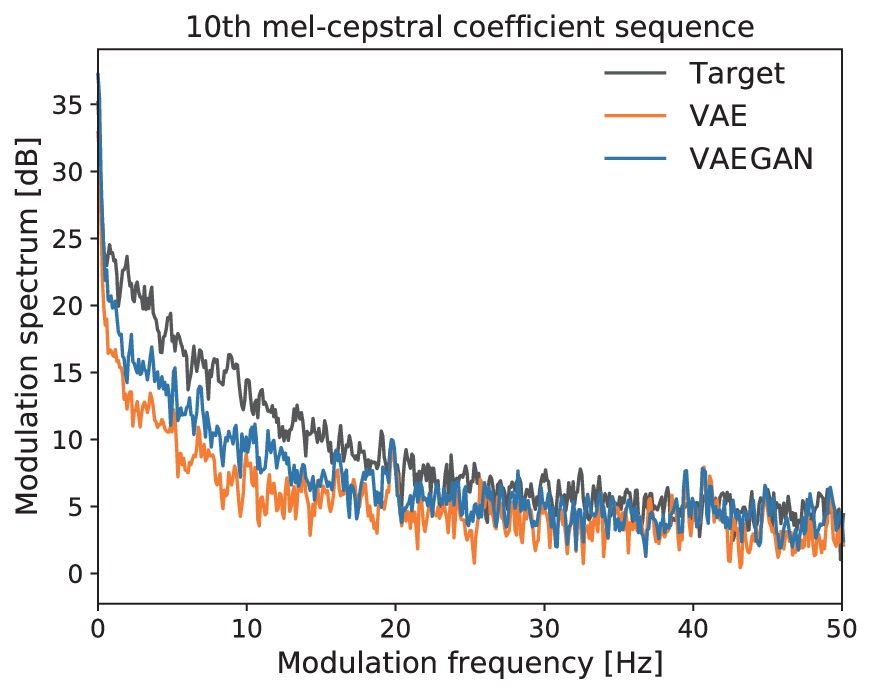}}
  \medskip
  \centerline{\includegraphics[width=.98\linewidth]{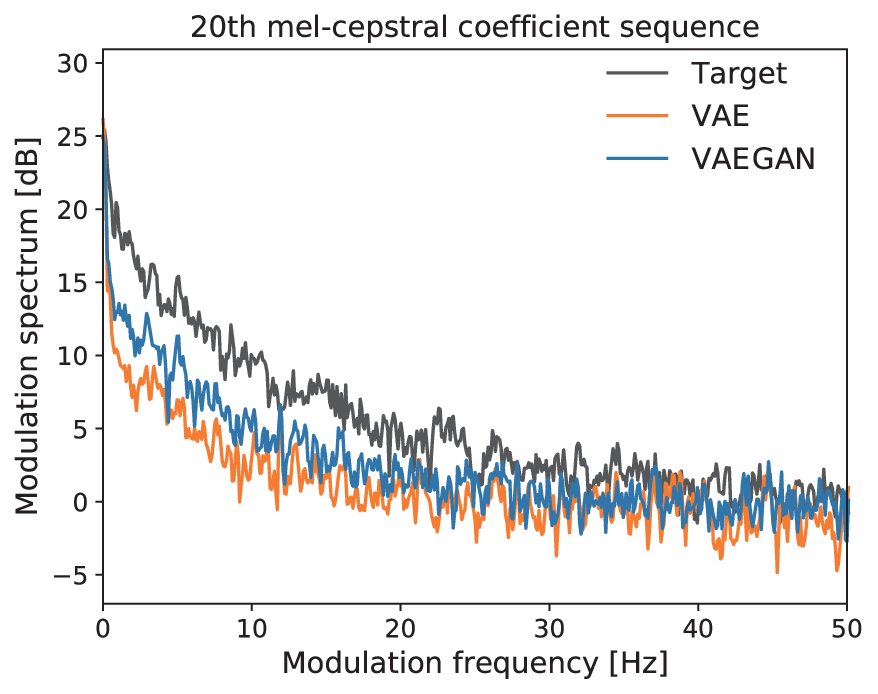}}
\end{minipage}
\begin{minipage}[t]{.195\linewidth}
  \centerline{\includegraphics[width=.98\linewidth]{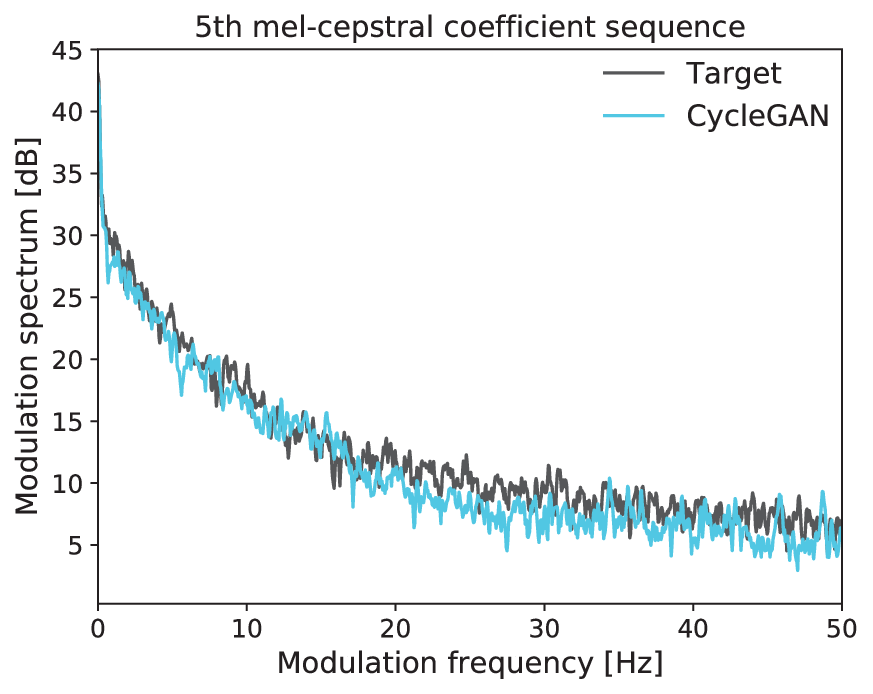}}
  \medskip
  \centerline{\includegraphics[width=.98\linewidth]{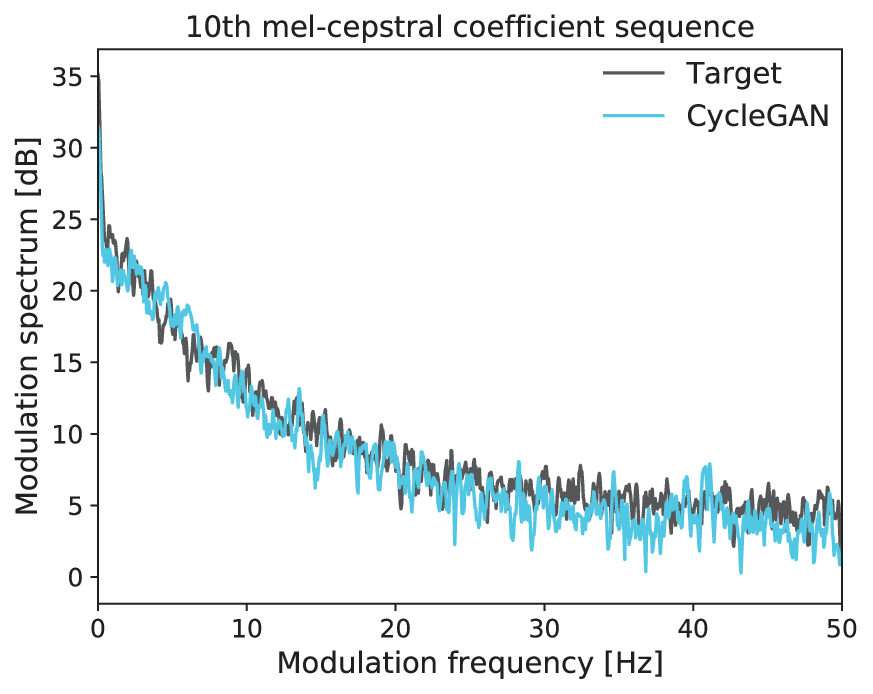}}
  \medskip
  \centerline{\includegraphics[width=.98\linewidth]{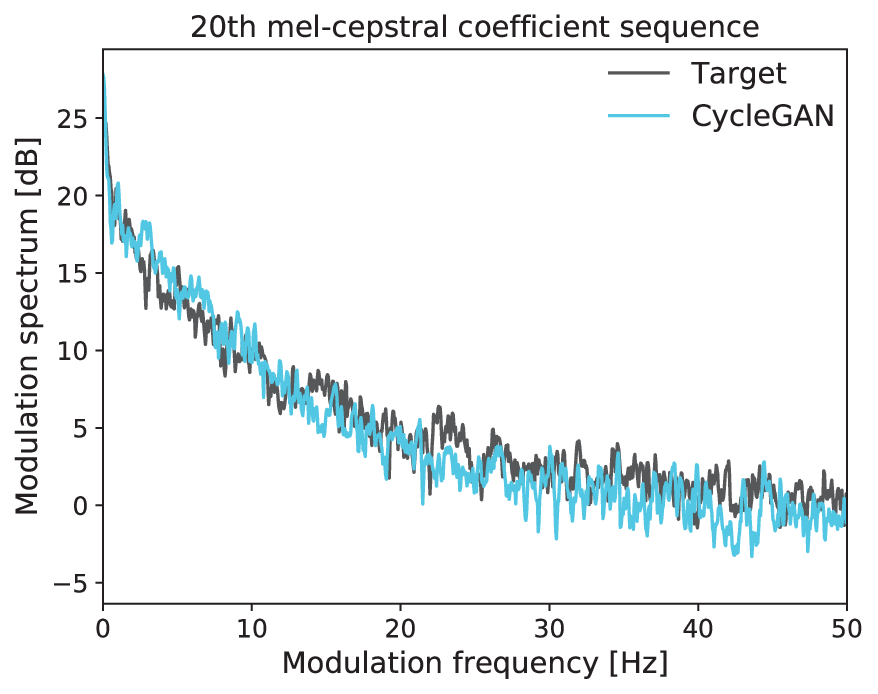}}
\end{minipage}
\begin{minipage}[t]{.195\linewidth}
  \centerline{\includegraphics[width=.98\linewidth]{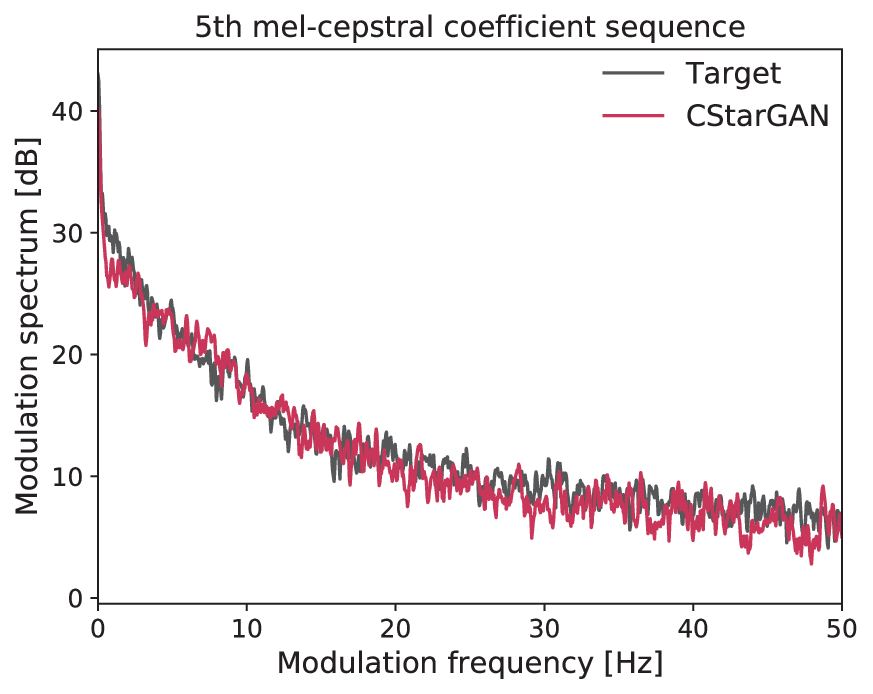}}
  \medskip
  \centerline{\includegraphics[width=.98\linewidth]{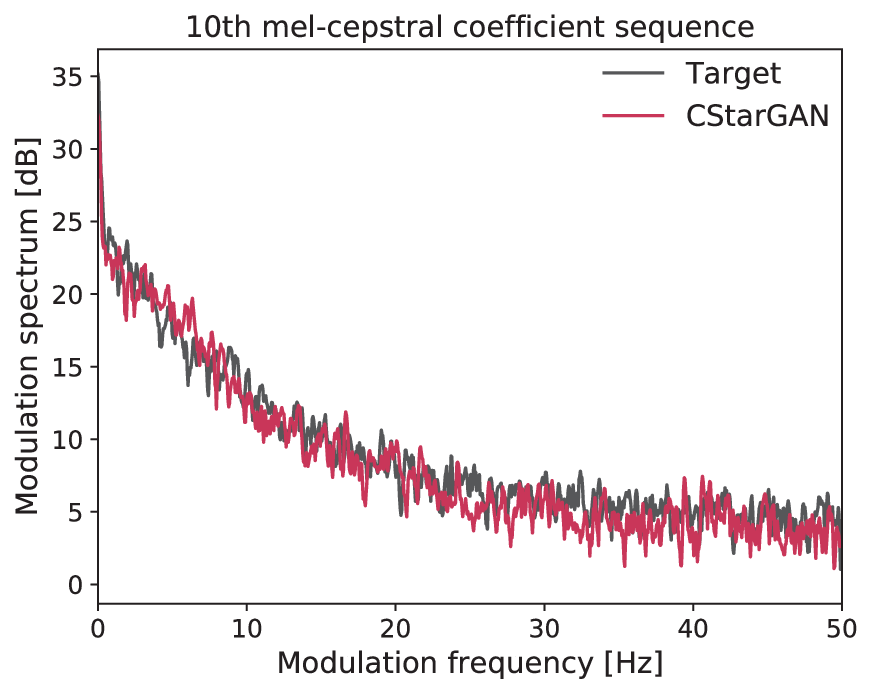}}
  \medskip
  \centerline{\includegraphics[width=.98\linewidth]{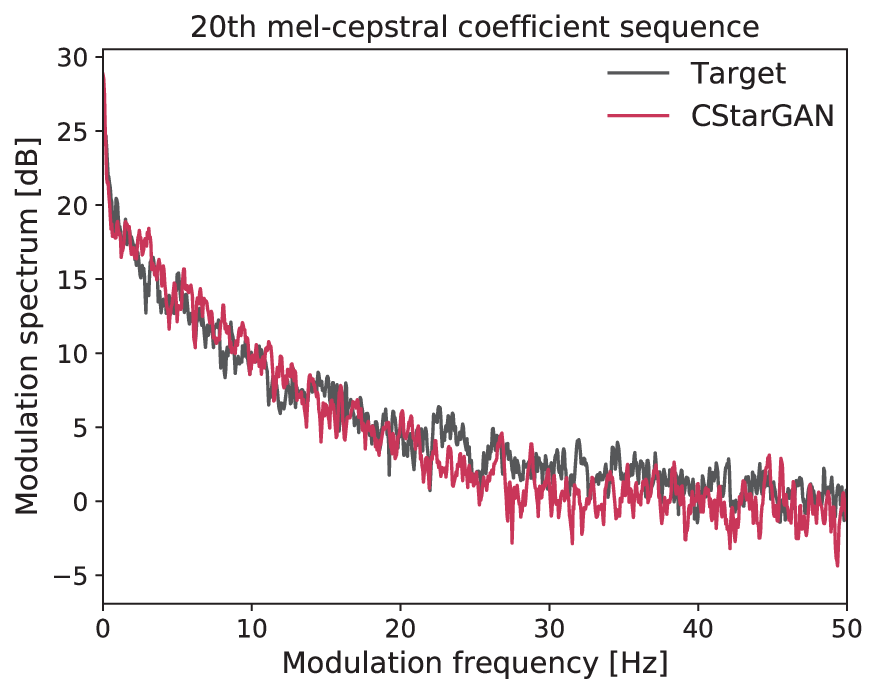}}
\end{minipage}
\centering
\begin{minipage}[t]{.195\linewidth}
  \centerline{\includegraphics[width=.98\linewidth]{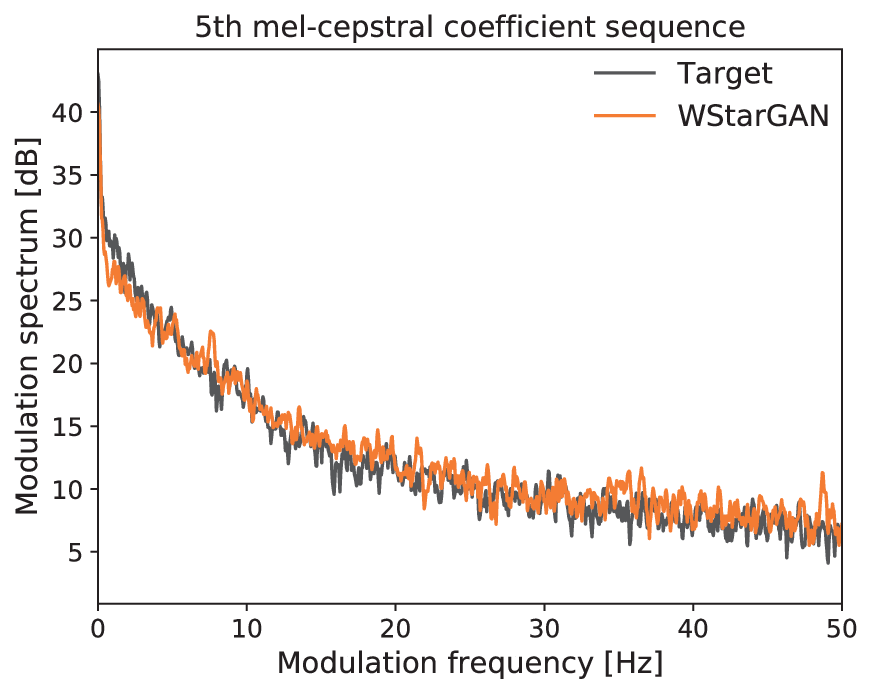}}
  \medskip
  \centerline{\includegraphics[width=.98\linewidth]{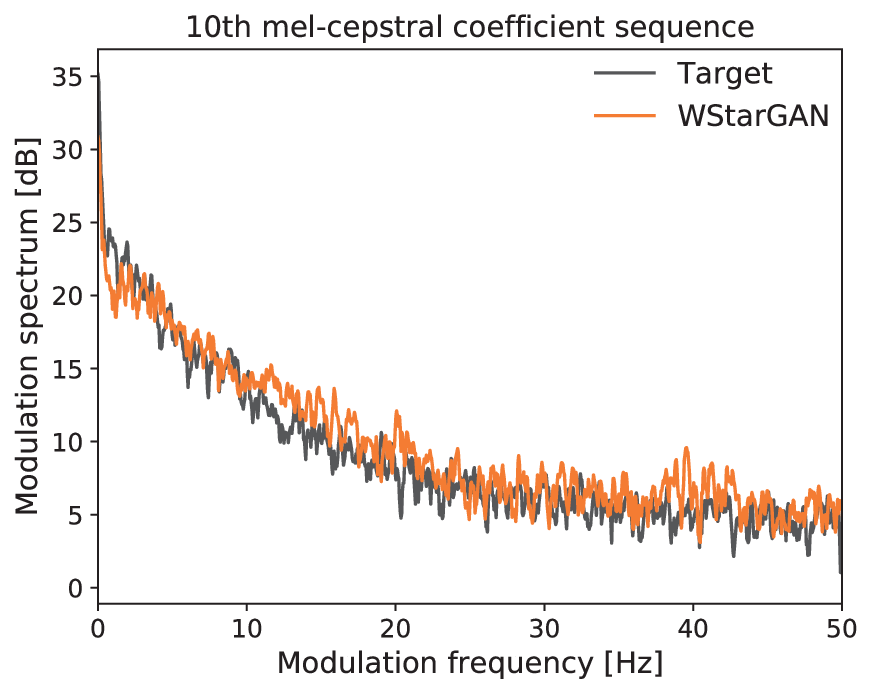}}
  \medskip
  \centerline{\includegraphics[width=.98\linewidth]{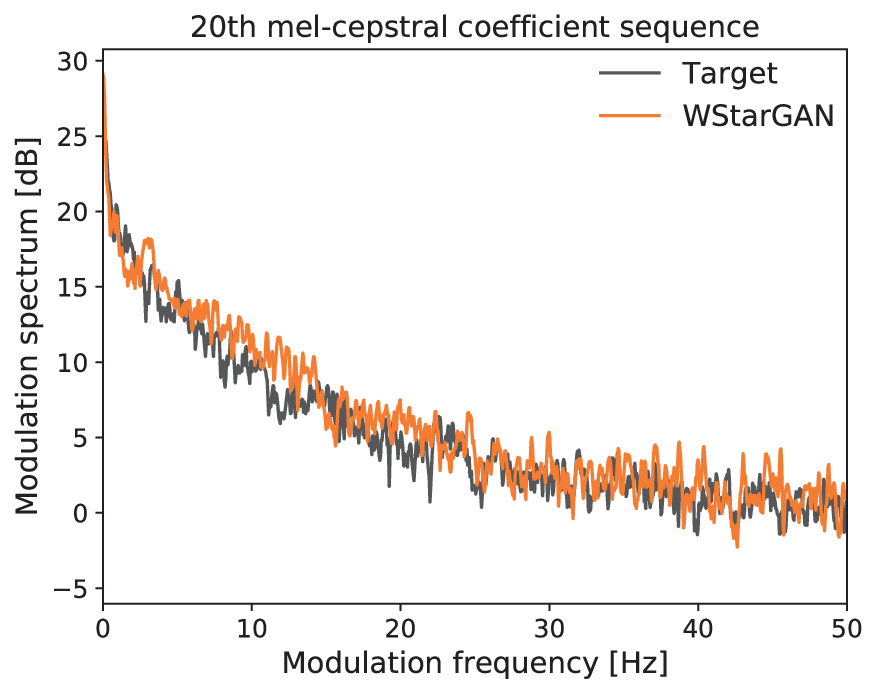}}
\end{minipage}
\begin{minipage}[t]{.195\linewidth}
  \centerline{\includegraphics[width=.98\linewidth]{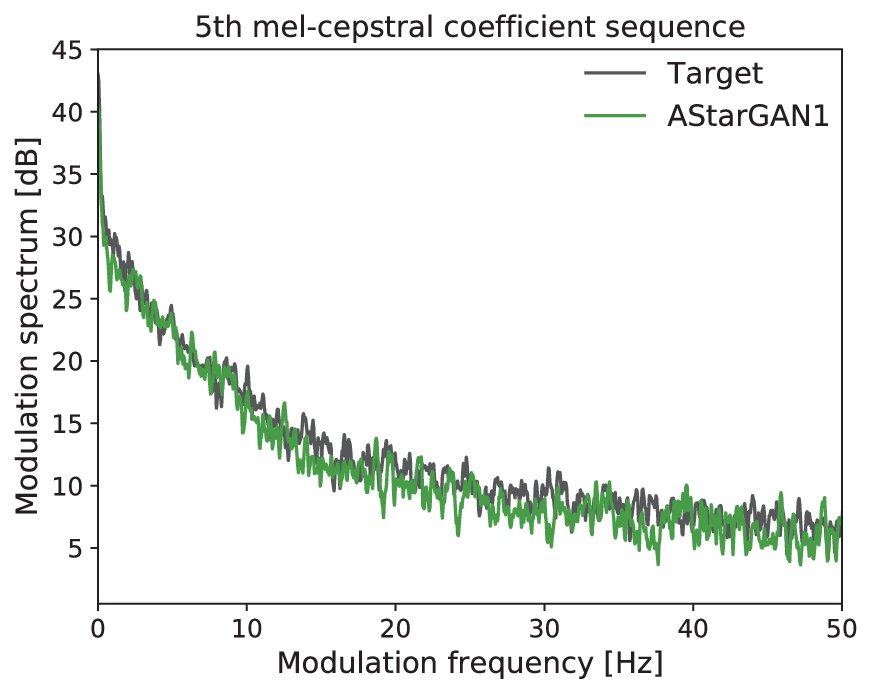}}
  \medskip
  \centerline{\includegraphics[width=.98\linewidth]{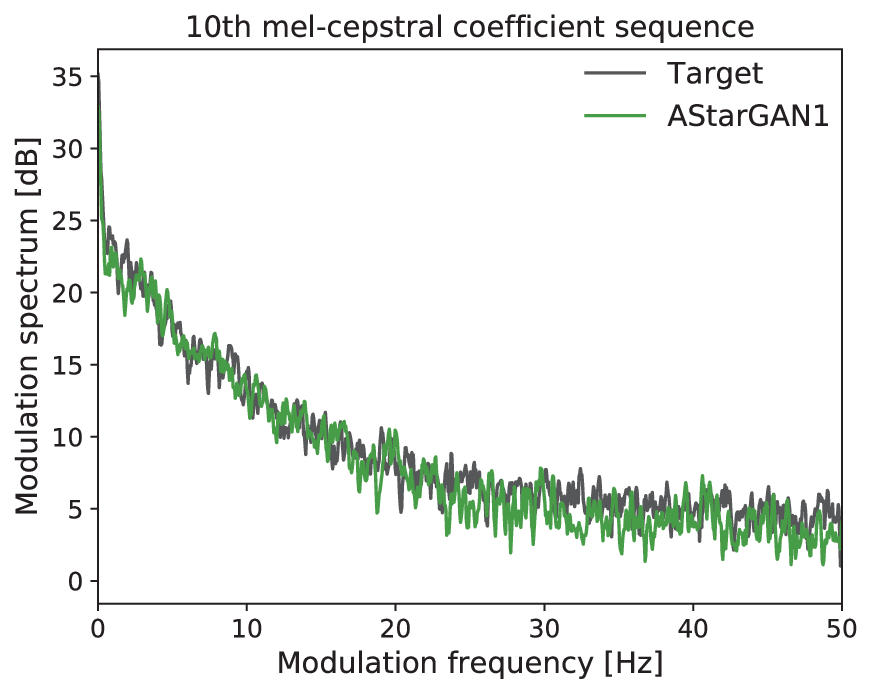}}
  \medskip
  \centerline{\includegraphics[width=.98\linewidth]{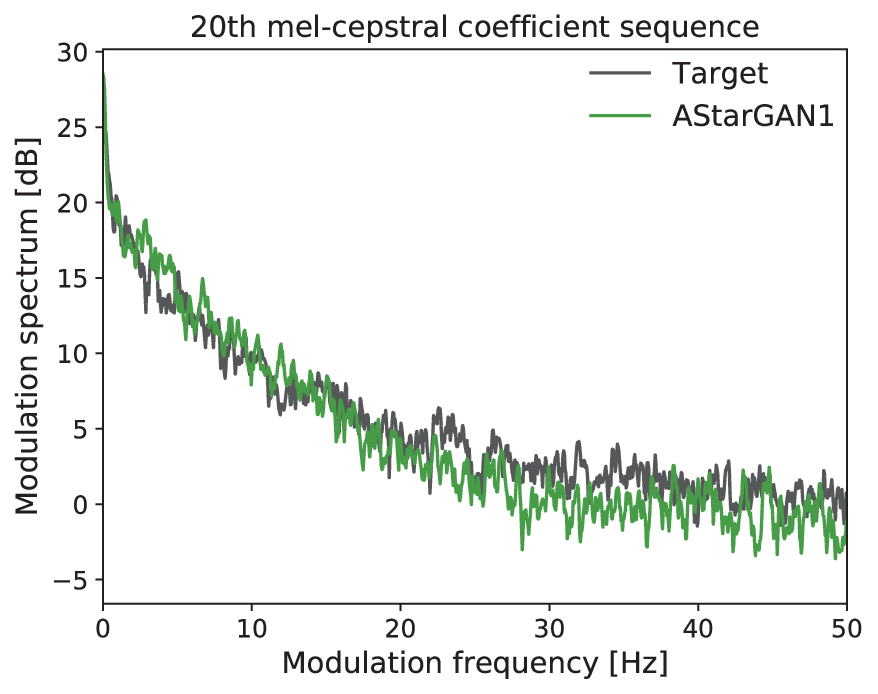}}
\end{minipage}
\caption{Average modulation spectra of the 5th, 10th and 20th dimensions of the converted MCC sequences obtained with the baseline methods and
the StarGAN-based methods.}
\label{fig:modspec_all}
\end{figure*}

\subsection{Subjective Listening Tests}

\begin{figure*}[t!]
\centering
\begin{minipage}[t]{.49\linewidth}
\centering
\centerline{\includegraphics[width=.98\linewidth]{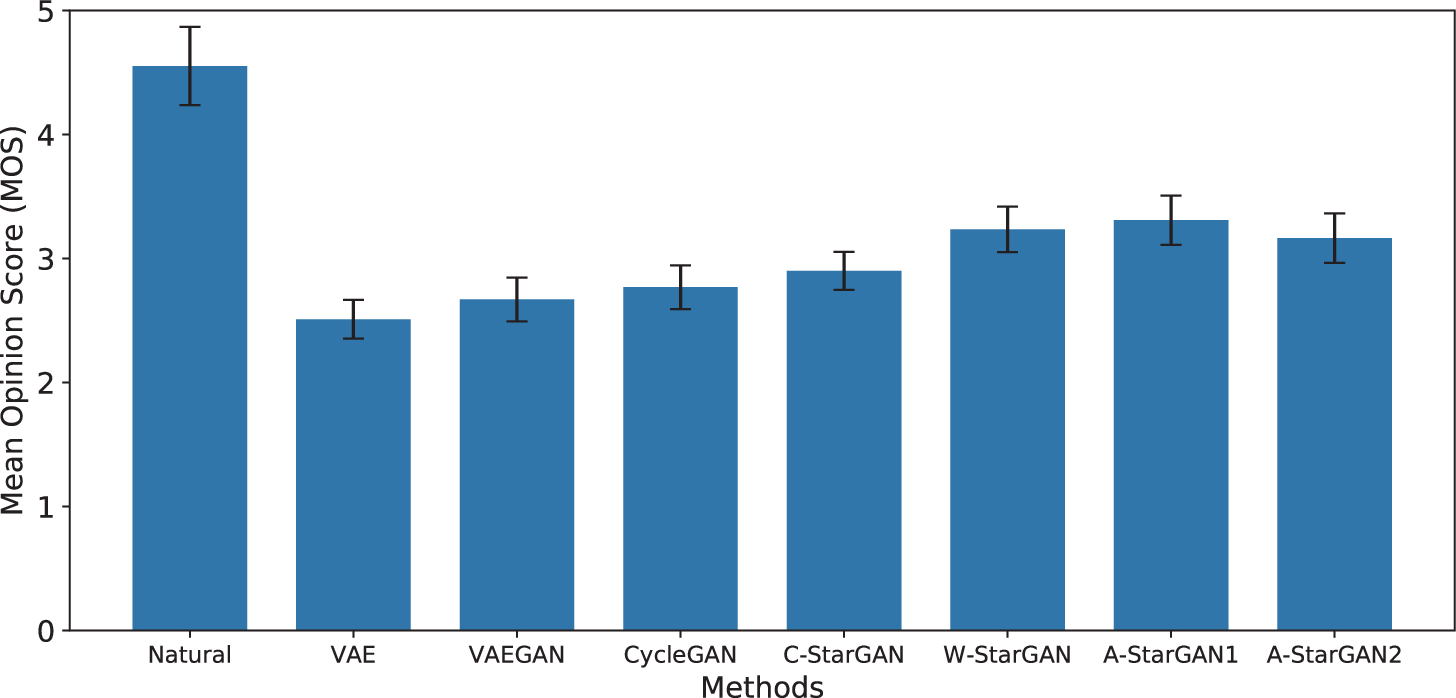}}
  \vspace{-1ex}
  \caption{Results of MOS test for speech quality}
  \label{fig:MOS_qlt}
\end{minipage}
\centering
\begin{minipage}[t]{.49\linewidth}
\centering
\centerline{\includegraphics[width=.98\linewidth]{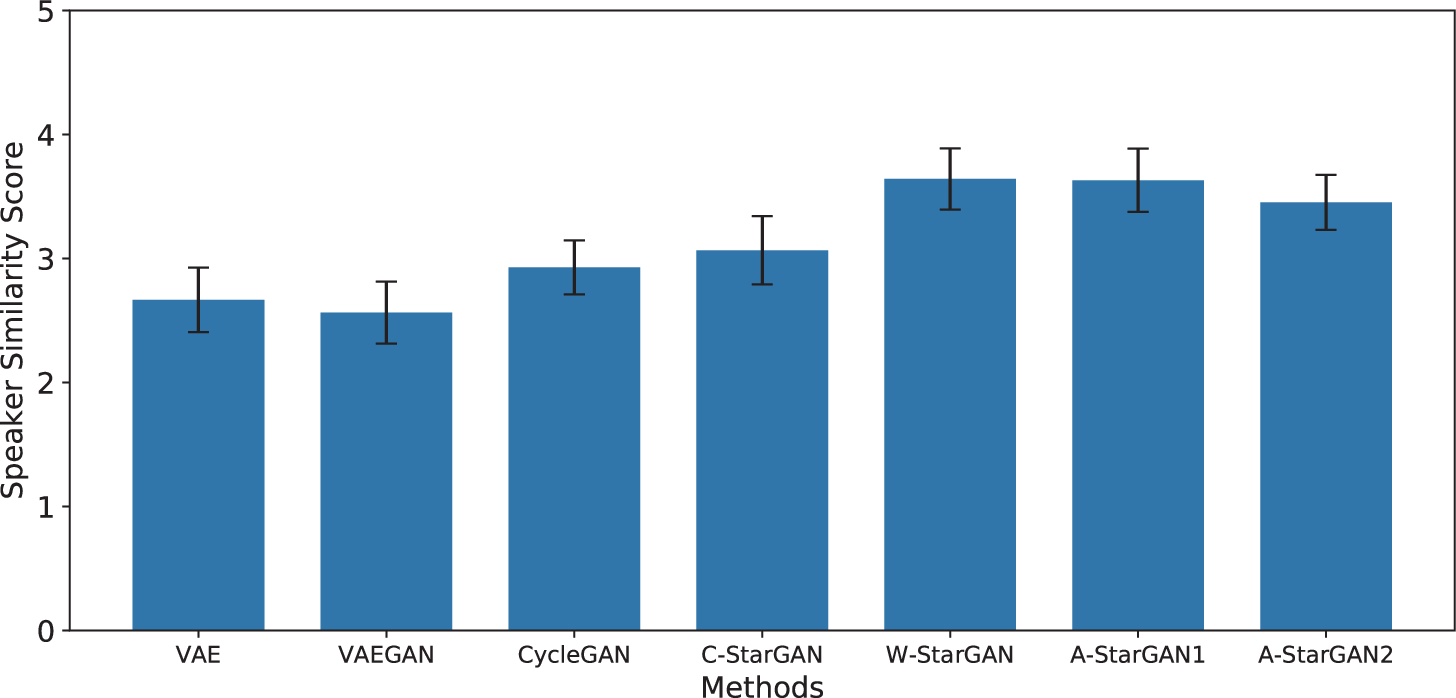}}
  \vspace{-1ex}
  \caption{Results of subjective speaker similarity test}
  \label{fig:MOS_sim}
\end{minipage}
\end{figure*}

\begin{figure}
\centering
\begin{minipage}[t]{.6\linewidth}
\centering
\centerline{\includegraphics[width=.98\linewidth]{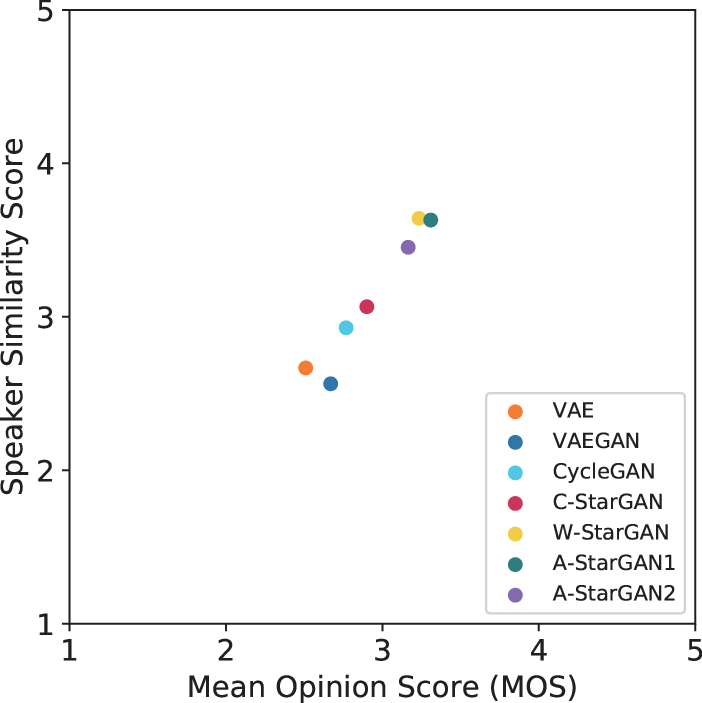}}
\caption{Scatter plot of \reffigs{MOS_qlt}{MOS_sim}}
\label{fig:scatter}
\end{minipage}
\end{figure}

We conducted 
subjective listening tests to
compare the speech quality and speaker similarity
of the converted speech samples
obtained with the proposed and baseline methods.
For these tests, we used the CMU ARCTIC dataset.
Twenty-four listeners (including 21 native Japanese speakers) 
participated in both tests.
The tests were conducted online, where each participant was asked 
to use a headphone in a quiet environment.

With the speech quality test, we evaluated the mean opinion score (MOS) for each speech sample.
In this test, we included the speech samples 
synthesized in the same way as the proposed and baseline methods (namely with the WORLD synthesizer) using
the acoustic features directly extracted from real speech samples. 
Hence, the scores of these samples are expected to show the upper limit of the performance. 
Speech samples were presented in random orders to eliminate bias as regards the order of the stimuli. 
Each listener was  
asked to evaluate the naturalness
by selecting 5: Excellent, 4: Good, 3: Fair, 2: Poor, or 1: Bad for each utterance.
The obtained scores
with 95\% confidence intervals are shown in \reffig{MOS_qlt}.
As the results show, 
A-StarGAN1 performed slightly better than 
W-StarGAN and A-StarGAN2 (although the differences were not significant)
and significantly better than C-StarGAN and the VAE and VAEGAN methods.
However, it also became clear that the speech quality obtained with all the 
methods tested here
was still perceptually distinguishable from real speech samples.

With the speaker similarity test, 
the subjective score for each sample was rated
on a five-point scale, as with the speech quality test. 
Each listener was given a converted speech sample and 
a real speech sample of the corresponding target speaker 
and asked to evaluate how likely they were to have been produced by the same speaker by selecting 
5: Definitely, 4: Likely, 3: Fairly likely, 2: Not very likely, or 1: Unlikely. 
The obtained scores
with 95\% confidence intervals
are shown in \reffig{MOS_sim}.
A scatter plot of the speech quality and speaker similarity scores 
obtained with the tested methods is shown in \reffig{scatter}.
As can be seen from the results,
the W-StarGAN and A-StarGAN formulations
performed comparably to each other and
showed significantly
better conversion ability 
than the remaining four methods.

Voice conversion examples 
are provided at \cite{Kameoka2020url_astargan-vc}.

\section{Conclusion}

In this paper, we proposed a method that allows nonparallel multi-domain VC 
based on StarGAN. We described three formulations of StarGAN and
compared them and several baseline methods 
in a nonparallel speaker identity conversion task.
Through objective evaluations, we confirmed that
our method was able to convert speaker identities reasonably well using
only several minutes of training examples.
Interested readers are referred to \cite{Kaneko2019short_cycleganvc2,Kaneko2019short_starganvc2}
for our investigations of other
network architecture designs and improved techniques for CycleGAN-VC and StarGAN-VC.

One limitation of the proposed method is that 
it can only convert input speech to the voice of a speaker 
seen in a given training set.
This is due to the fact that  
one-hot encoding (or a simple embedding) used for speaker conditioning is 
nongeneralizable to unseen speakers.
An interesting topic for future work includes 
developing a zero-shot VC system that can convert input speech to 
the voice of an unseen speaker by looking at only a few of his/her utterances.
As in the recent work \cite{Qian2019},
one possible way to achieve this involves using 
a speaker embedding 
pretrained based on a metric learning framework
for speaker conditioning.


\ifCLASSOPTIONcaptionsoff
  \newpage
\fi



\bibliographystyle{IEEEtran}
\bibliography{Kameoka2018arXiv11}

\end{document}